\documentclass[aps,pra,reprint,floatfix,longbibliography]{revtex4} 

\usepackage{graphicx} 
\usepackage{amsmath, amsthm, amssymb}
\usepackage[utf8x]{inputenc}
\usepackage{color}

\usepackage{siunitx}
\usepackage{booktabs}

\usepackage[x11names, rgb]{xcolor}

\usepackage{tikz}
\usetikzlibrary{shadows}
\usetikzlibrary{calc,arrows,shapes}
\usetikzlibrary{decorations.markings, decorations.pathreplacing} 

\tikzset{%
        brace/.style = { decorate, decoration={brace, amplitude=5pt} },
       mbrace/.style = { decorate, decoration={brace, amplitude=5pt, mirror} }
}

\newtheorem{thm}{Theorem}[section]
\newtheorem{prop}[thm]{Proposition}
\newtheorem{con}[thm]{Conclusion}

\newtheorem{dfnsz}[thm]{Definition and Proposition}
\newtheorem{lem}[thm]{Lemma} \newtheorem{cor}[thm]{Corollary}

\theoremstyle{remark}
\newtheorem{rmk}[thm]{Remark}
 \newtheorem{xmpl}[thm]{Example}

\theoremstyle{definition} 
\newtheorem{dfn}[thm]{Definition}

\newtheorem{conv}[thm]{Convention}

\begin{document}

\title{On emittance and optics calculation from the tracking data in periodic lattices}

\author{Malte Titze}
\affiliation{CERN, Geneva, Switzerland}


\begin{abstract}
In this work we examine the interplay between normal form and matched particle distributions in a linear setting. We first outline the connection between the established \textit{$\Sigma$-matrix method} and Williamson's Theorem. Then we show that the Iwasawa decomposition provides a natural framework for a description of beam optics parameters. Along the way we will apply these methods to a realistic tracking example, as well as provide additional examples, including the connection to the parameterization of Courant-Synder.
\end{abstract}

\maketitle


\section{Introduction}

In a particle simulation involving a periodic lattice, it is usually desired to generate particles in a \textit{matched} state, which means that the shape of the distribution should not change after one passage through the lattice. In fact, if a matched distribution can be found, one often has already accomplished a great deal in the understanding of the simulation. Additionally, there are circumstances in which the knowledge of the effective emittance and optics parameters is required but difficult to compute, for example in the case of the MAD-X space charge module for the CERN PS lattice (without poleface windings) near the integer resonance $Q_x = 6$, where the search for a closed-orbit can break down.\\\\
In recent years, efforts have been made with success in using covariance matrices to compute emittances and beam optics parameters. This was demonstrated in Ref. \cite{bib:lebedev1999} in the 4D case and in \cite{bib:wolski2006, bib:nash2006} for the 6D case. In a later published version of \cite{bib:lebedev1999}, i.e. in Ref. \cite{bib:lebedev2010}, it was used to examine the relation between the Edwards-Teng \cite{bib:edwards1973} (see also Ref. \cite{bib:sagan1999} for a summary) and the Mais-Ripken \cite{bib:willeke1988} parameterization in a 4D situation with coupled optics. The results were picked up in \cite{bib:alexahin2013} in order to compute the emittances from the covariance matrices of 6D tracking data and, in regards of code implementation, recent progress has been made to include some of these techniques into the MAD-X space charge module \cite{bib:alex2018}.
In Refs. \cite{bib:luo2004, bib:qin2009, bib:qinetal2013} generalizations of the Courant-Snyder parameterization to 4D were examined.\\\\
The aim of this work is to continue in this spirit by systematically exploring the connection to linear normal form and established theorems regarding symplectic matrices: By utilizing Williamson's Theorem we obtain a proof of the remarkable result to obtain emittances by symplectic diagonalization of the given covariance matrix. Such diagonalizations are not unique but, as we shall see in Sec. \ref{subs:param}, by knowledge of how the underlying freedom enters into the equations we outline how to obtain faithful optics information out of the tracking data.\\\\
In particular we found that the Iwasawa decomposition provides a natural framework: Two of their three factors always remain the same, while the third factor can be determined under one additional condition: Namely, that the emittances have to be mutually distinguishable. We are thus led, rather naturally, to a characterization of e.g. the optics $\beta$-functions. We also discuss an alternative route to obtain $\beta$-functions by a statistical argument, which was proposed in Ref. \cite{bib:wolski2006}, and which we connect and apply to our situation.\\\\
Along the way we will provide several tables from a realistic tracking scenario and three examples which illustrate the results of the technical steps. The first example establishes the connection to the well-known (2D) Courant-Snyder parameterization. The second example contains a short way of how to obtain the emittances from a 4D covariance matrix, without the necessity to compute the eigenvectors. The last example also deals with a 4D situation, now with a single coupling parameter. By means of this last example we will demonstrate certain properties of the general decomposition.\\\\
For practical purposes we will summarize in Sec. \ref{subsec:emit_cov} the techniques of how emittance calculations can be performed in a linear scenario, establish the connection to the familiar emittance of Lapostolle and the single-particle action. For completeness we also discuss the situation of measuring beam sizes at three different locations in the lattice.

\section{On invariant covariance matrices} \label{subs:param}

\subsection{Motivation and preliminaries} \label{subs:prelims}

Consider a tracking simulation which produces, at every turn, a
distribution of particles depending on an initial set of coordinates. We can compute the moments of these distributions in phase space and obtain some sort of measure of the phase space volume occupied. It is of great interest to understand how to set up a distribution in which certain functions of these macroscopic quantities remain unchanged or vary only very slowly in the course of the simulation. In the following we will understand our lattice to be in the form of a ring, but the same reasoning can be applied to a straight periodic lattice.\\\\
Let $\mathcal{F} \colon \mathcal{P} \to \mathcal{P}$ be a canonical transformation from phase space $\mathcal{P} \subset \mathbb{R}^{2n}$ onto itself, which describes the physics of the storage ring in form of a single turn around the machine at a given fixed position in the ring. Such a \textit{one-turn} or \textit{Poincar\'e} map is usually the result of a composition of many elementary maps, which describe the individual elements of the machine.\\\\
In this work we will examine the situation in the vicinity of an assumed closed orbit, where linear effects play the dominant role. Therefore we will be focusing on the first derivative $M$ of the one-turn map $\mathcal{F}$ at the closed orbit and do not consider any higher-order effects of the full map $\mathcal{F}$. Because of this restriction -- and for brevity -- we will also call $M$ the one-turn map. This map is symplectic since $\mathcal{F}$ is canonical\footnote{Here 'canonical' means a symplectomorphism.}.\\\\
If $g \colon \mathcal{P} \to [0, 1]$ denotes the phase space density of a particle distribution,
its covariance matrix $G$, consisting of the second-order moments, is given by
\begin{equation}
G := \sum_{k, l} \langle x_k x_l \rangle \, e_k e_l^{tr} = \langle x x^{tr} \rangle, \label{eq:g_cormat}
\end{equation}
where, for any integrable function $h \colon P \to \mathbb{R}$, the mean $\langle h \rangle$ is given by
$\langle h \rangle := \int g(x) h(x) dx .$
We see that $M$ acts by
matrix congruence on $G$, where the new covariance matrix $G'$ is given by
\begin{equation}
G' = \langle M x (Mx)^{tr} \rangle = \langle M x x^{tr} M^{tr} \rangle = M G M^{tr} . \label{eq:congruence_inv}
\end{equation}
Such covariance matrices are important, because their entries are the ingredients to compute the emittances of the beam; in the 2D case:
\begin{equation}
\det \left( \begin{array}{cc} \langle x^2 \rangle & \langle x p_x \rangle \\
\langle x p_x \rangle & \langle p_x^2 \rangle \end{array} \right) = \langle x^2 \rangle \langle p_x^2 \rangle
- \langle x p_x \rangle^2 = \epsilon_x^2 . \label{eq:2d_emittance}
\end{equation}
Moreover, by means of Eq. \eqref{eq:congruence_inv}, we have a way to follow the evolution of the moments in the course of the tracking. Because $\det(M) = 1$ for symplectic maps, the emittance in Eq. \eqref{eq:2d_emittance} is conserved. Note that the emittance is just one example of an invariant. In Ref. \cite{bib:rangarajan89} functions of higher-order moments which remain invariant with respect to symplectic matrices were studied.\\\\
Here we are focusing on second-order moments and address as our first goal the following question: Given $M$, how can we classifying all 'matched' cases in which $G' = G$ holds? As we shall see in the course of this section, the answer will connect a property used in e.g. Ref. \cite{bib:nash2006} to linear normal form.
\begin{dfn} \label{dfn:m_inv}
We say a matrix $G \in \mathbb{R}^{2n \times 2n}$ is \textit{$M$-congruent invariant}
or, for brevity, \textit{$M$-invariant}, if $M G M^{tr} = G$ holds. 
\end{dfn} \noindent
To begin with, we recall an important fact which we
will frequently use to identify covariance matrices.
\begin{thm} \label{thm:cov_spsd}
$G \in \mathbb{R}^{m \times m}$ is a covariance matrix if and only if $G$ is symmetric and positive semidefinite.
\end{thm}
\begin{proof}
A proof for convenient reference is included in Appendix \ref{app:ec}.
\end{proof} \noindent
This means we are interested in $M$-invariant symmetric and positive semidefinite matrices $G$, having in mind that $M$ is the symplectic one-turn map of a lattice in a particle accelerator. In particular this means that we can assume that the complexification $M_\mathbb{C}$ of $M$ can be diagonalized with mutually distinguishable eigenvalues (tunes).\\\\
We will now systematically develop important properties of $M$-invariant symmetric matrices $G$. Throughout this section convention \ref{conv:ab} will hold and all matrices are given with respect to the ordering $x, y, z, p_x, p_y, p_z$ unless otherwise stated.

\subsection{Invariance and linear normal form} \label{subs:normal_form}

In this paragraph we outline the interplay between the technique of transforming a given 
linear symplectic map $M$ with mutually distinguishable eigenvalues into a normal form and invariant covariance matrices. We will use an important lemma which not only helps us to describe the connection, but also in the next paragraph \ref{ssubs:class} where we examine the degree of freedom involved in the matrices. For a procedure of how to construct the normal version of a general (higher-order) one-turn map we refer the reader to Refs. \cite{bib:forest1990, bib:bengtsson1990}. Preliminary tools are given in Appendix \ref{app:ec}. 

\begin{thm}[Linear normal form] \label{thm:s_diag_r}
Let $M \in \mathrm{Sp}(2n; \mathbb{R})$ diagonalizable with mutually distinguishable eigenvalues
on the unit circle. Then there exists $V \in \mathrm{Sp}(2n; \mathbb{R})$, 
so that $R := V^{-1} M V$ is orthogonal, leaving the plane
$E_k := \mathrm{span}\{e_k, e_{n + k}\}$ for $k \in \overline{n}$ invariant:
\[
\forall k \in \overline n \colon \;\; R |_{E_k} =
\left(
\begin{array}{cc}
\cos(\varphi_k) & \sin(\varphi_k) \\
- \sin(\varphi_k) & \cos(\varphi_k)
\end{array}
\right) ,
\]
where the phases $\varphi_k \in [0, 2 \pi[$ are related to the eigenvalues 
$\lambda_k$ of $M$ by $\lambda_{j_k} = \exp(i \varphi_k)$, and where the $j_k$'s correspond to a
representation introduced in Dfn. \ref{dfn:equirel}, so that for the corresponding eigenvectors $a_{j_k}$ of $M$,
$\langle a_{j_k}, J a_{j_k} \rangle$
has positive imaginary part.
\end{thm} \noindent
With the operator $\diamond$ introduced in
Def. \ref{dfnsz:oplus} we can conveniently write $R = R_1 \diamond \cdots \diamond R_n$. In the accelerator-physics terminology the phase space in which the one-turn map has the above \textit{normal} form is also called \textit{Floquet-space}.
\begin{proof}
A proof can be found in Ref. \cite{bib:dragt2018} or in Appendix \ref{app:ec}.
\end{proof} \noindent
The matrix $V$ in Thm. \ref{thm:s_diag_r} is not unique, as can readily be seen if composing the map by additional rotations (which are symplectic) leaving the individual planes invariant:
\begin{rmk} \label{rmk:freedom_v}
By lemma \ref{lem:r_invariance} we can conclude that the linear normal form map $V$ of Thm. \ref{thm:s_diag_r} is determined up to an orthosymplectic transformation on the left; namely if $V^{-1} M V = R = \tilde V^{-1} M \tilde V$, i.e. $V R V^{-1} = \tilde V R \tilde V^{-1}$, then $\tilde V^{-1} V = D_1 + J_2^{\oplus n} D_2$. Because $V$ and $\tilde V$ are symplectic, it follows that $D_1^2 + D_2^2 = 1$ and so $\tilde V^{-1} V$ has the same form as $R$.
\end{rmk}
\begin{thm} \label{thm:m_invariance}
Let $M \in \mathrm{Sp}(2n; \mathbb{R})$ diagonalizable with mutually distinguishable eigenvalues on the unit circle and $G$ symmetric.
Then $M G M^{tr} = G$ if and only if there exists a diagonal matrix
$D = \mathrm{diag}(\Lambda, \Lambda)$ with $\Lambda = \mathrm{diag}(b_1, ..., b_n)$ so that
$G = V D V^{tr}$, where $V \in \mathrm{Sp}(2n; \mathbb{R})$ is given according to Thm. \ref{thm:s_diag_r}.
\end{thm}
\begin{proof}
Let $T$ be the orthogonal operator discussed in Rmk. \ref{rmk:dia_oplus} and $R = V^{-1} M V$ according to 
Thm. \ref{thm:s_diag_r}.
Then we have $\tilde R := R_1 \oplus \cdots \oplus R_n = T^{tr} R T = T^{tr} V^{-1} M V T$, 
and $G$ is $M$-invariant if and only if
$D_0 := T^{tr} V^{-1} G V^{-tr} T$ is $\tilde R$-invariant:
\[
\tilde R D_0 \tilde R^{tr} = T^{tr} V^{-1} M V T T^{tr} V^{-1} G V^{-tr} T T^{tr} V^{tr} M^{tr} V^{-tr} T = 
T^{tr} V^{-1} M G M^{tr} V^{-tr} T = D_0 ,
\]
and by lemma \ref{lem:r_invariance} this the case if and only if $D_0$ has the form
$D_0 = \mathrm{diag}(b_1, b_1, b_2, b_2, ..., b_n, b_n)$, which is equivalent to $D := T D_0 T^{tr} = V^{-1} G V^{-tr}$
has the form $D = \mathrm{diag}(\Lambda, \Lambda)$ with $\Lambda = \mathrm{diag}(b_1, b_2, ..., b_n)$.  
\end{proof}
This result can be connected to a statement used in Ref. \cite{bib:nash2006}, and to obtain a familiar expression in the form of eigenvector decompositions, as follows:
\begin{thm} \label{cor:real_sym_inv_basis}
Every real-valued symmetric invariant $G$ of $M$, where $M$ is diagonalizable with mutually distinguishable
eigenvalues, can be represented as a sum $G = \sum_{k = 1}^n g_k Z_k$ of $n$ elementary matrices with $g_k \in \mathbb{R}$, where the $Z_k$ are given by $Z_k = a_{j_k} a_{j_k}^H + \bar a_{j_k} a_{j_k}^{tr}$, and where $\{j_1, j_2, ..., j_n\} \subset \overline{2n}$
is a representation system according to Conv. \ref{conv:ab}.
\end{thm}
\begin{proof}
The statement can be found in e.g. Ref. \cite{bib:nash2006} in a slightly different version. A proof was included in Appendix \ref{subs:symbasis}. In \ref{subs:symbasis} -- and only there -- we have changed our notation to $M^{tr}$. This means, by Conv. \ref{conv:ab}, that the eigenvectors $a_{j_k}$ appearing there are proportional to the vectors $J a_{j_k}$ here. 
\end{proof}
If we recall that by construction $V(e_k + i e_{n + k}) = a_{j_k}$ for $k \in \bar n$ holds (see the proof of Thm. \ref{thm:s_diag_r}), then we see that Thm. \ref{cor:real_sym_inv_basis} is equivalent to Thm. \ref{thm:m_invariance}, which was obtained in a rather different manner:
\begin{equation}
V D V^{tr} = G = \sum_{k = 1}^n g_k Z_k \;\; \Leftrightarrow \;\; D = \sum_{k = 1}^n g_k V^{-1}(a_{j_k} a_{j_k}^H + \bar a_{j_k} a_{j_k}^{tr}) V^{-tr} = \sum_{k = 1}^n 2 g_k (e_k e_k^{tr} + e_{n + k} e_{n + k}^{tr}) .  
\end{equation}
In some sense Thm. \ref{thm:m_invariance} lays at the heart of computing emittances (i.e. the entries of $D$) out of covariance matrices using linear normal forms and therefore in answering the questions raised in the introduction of this section. We will now turn our attention to the emittance and optics computation.

\subsection{Classification of invariant covariance matrices} \label{ssubs:class}

The matrix $G$ in Thm. \ref{thm:m_invariance} was only assumed to be symmetric. In particular this includes our case, where $G$ comes from a covariance matrix of a particle distribution. By our remark in paragraph \ref{subs:prelims}, these matrices are additionally positive semidefinite. In the typical situation of tracking the distribution through the accelerator, the beam will not be degenerated, i.e. the diagonal entries of $D$ in Thm. \ref{thm:m_invariance}, which correspond to the emittances, as we shall see in Sec. \ref{subsec:emit_cov}, are always positive. This means that we can well assume that $G$ is positive definite.\\\\
In this case even more can be said about such invariants: Thm. \ref{thm:m_invariance} effectively makes a statement on the conditions by which $G$ can be diagonalized by matrix congruence via a symplectic map $V$. Since $V$ was constructed via $M$ by Thm. \ref{thm:s_diag_r}, the linear optics of the machine is -- up to the tune -- effectively contained in $V$. On the other hand, an abstract symplectic diagonalization of $G$ without knowledge of the optics is always possible in form of Williamson's Theorem:

\begin{thm} [Williamson \cite{bib:william1936, bib:william1937}] \label{thm:williamson}
Let $G$ be a $2n$-dimensional real symmetric positive definite matrix. Then there exist $S \in \mathrm{Sp}(2n; \mathbb{R})$ so that 
\begin{equation}
G = S^{tr} D S  \label{eq:gstds}
\end{equation}
with $D = \mathrm{diag}(\Lambda, \Lambda)$ and $\Lambda = \mathrm{diag}(\lambda_1, ..., \lambda_n)$.
\end{thm}
\begin{proof}
 A concise proof of this theorem can be found in Ref. \cite{bib:simon1998}. Because of its relevance in our computations we will sketch the proof here. Since $G$ is symmetric and positive definite, it can be diagonalized by an orthogonal matrix and all its eigenvalues are positive. Hence it admits an invertible square root $G^{1/2}$. Now anti-diagonalize the antisymmetric matrix $G^{-1/2} J G^{-1/2}$, i.e. find an orthogonal matrix $A$ so that
\begin{equation}
A^{tr} G^{-1/2} J G^{-1/2} A = \left(\begin{array}{cc} 0 & \Omega \\ -\Omega & 0 \end{array}\right)
\end{equation}
holds, where $\Omega$ is a diagonal $n \times n$-matrix with positive entries. Then set 
$S := \mathrm{diag}(\Omega^{1/2}, \Omega^{1/2}) A^{tr} G^{1/2}$ and verify symplecticity and Eq. \eqref{eq:gstds} (with $D = \Omega^{-1}$).
\end{proof}
The set of positive real quantities $\lambda_i > 0$ obtained in the above manner is known in the literature as the \textit{symplectic spectrum} of $G$ \cite{bib:degosson2006}. As already indicated in the above proof, and in particular by the next theorem, it will become apparent that the symplectic spectrum is $S$-independent. However, the symplectic matrices diagonalizing $G$ are not unique, as we shall see by the examples given below. But by the next classification theorem they are not 'too far away' from each other:
\begin{thm} \label{thm:mt}
Let $M \in \mathrm{Sp}(2n; \mathbb{R})$ be a symplectic matrix with mutually distinguishable eigenvalues and
$G_i$, $i = 1, 2$, symmetric and positive definite, so that
\begin{equation}
  M G_i M^{tr} = G_i  \label{eq:inv}
\end{equation}
hold.  Let $G_i = S_i^{tr} D_i S_i$ be symplectic diagonalizations
with $D_i = \mathrm{diag}(\Lambda_i, \Lambda_i)$. 
Then it holds $S_2 S_1^{-1} \in SO(2n; \mathbb{R})$.
\end{thm}
For the proof of this theorem we have dedicated two smaller paragraphs \ref{subsec:evev} and \ref{subsec:evev2} in the appendix.
One immediate implication together with Thm. \ref{thm:m_invariance} is the justification of the pretty remarkable result that one can find the emittances out of a covariance matrix alone, without having knowledge of the underlying optics given by the one-turn map. The only assumption on the covariance matrix is that it belongs to a matched distribution with respect to the unknown linear optics: Namely if $G$ is given with a decomposition $G = V D V^{tr}$ according to Thm. \ref{thm:m_invariance} and another decomposition $G = S^{tr} D_1 S$ has been found, for example from Thm. \ref{thm:williamson}, then we obtain
with the orthogonal $W := V^{tr} S^{-1}$:
\[
D_1 = S^{-tr} G S^{-1} = S^{-tr} V D V^{tr} S^{-1} = W^{tr} D W
\]
and since the eigenvalues can not be changed by orthogonal matrix congruence, $D_1$ and $D$ must contain the same entries up to a suitable permutation. In Sec. \ref{subsec:emit_cov} we will see how the diagonal entries of $D$ are connected to the classical emittances by Lapostolle (in the 2D case).\\\\
\begin{table}[h] \centering
\addtolength{\tabcolsep}{-0.4cm}
\hspace{-1.5cm}
\resizebox{\textwidth}{!}{
\begin{tabular}{SSSSSS}
6.179e-06 & 1.536e-08 & 1.329e-05 & 1.138e-07 & -8.148e-11 & -4.819e-07 \\
1.536e-08 & 2.342e-06 & 8.767e-08 & 4.230e-10 & -4.951e-08 & -1.188e-08 \\
1.329e-05 & 8.767e-08 & 2.922e-02 & 1.713e-07 & 5.298e-09 & 1.791e-06 \\
1.138e-07 & 4.230e-10 & 1.713e-07 & 3.084e-09 & -7.580e-13 & 5.042e-09 \\
-8.148e-11 & -4.951e-08 & 5.298e-09 & -7.580e-13 & 2.480e-09 & 3.380e-10 \\
-4.819e-07 & -1.188e-08 & 1.791e-06 & 5.042e-09 & 3.380e-10 & 1.371e-06 \\ 
\end{tabular}
}
\caption[Covariance matrix example during SPS tracking]{Example of a covariance matrix belonging to a stable particle beam in the SPS at $Q_x=20.144$.}
\label{tab:emit_g}
\end{table}
\begin{table}[h] \centering
\addtolength{\tabcolsep}{-0.4cm}
\hspace{-1.5cm}
\resizebox{\textwidth}{!}{
\begin{tabular}{SSSSSS}
7.121e-08 & -2.906e-11 & 4.643e-08 & -1.405e-10 & -1.657e-11 & 2.463e-08 \\
-2.906e-11 & 5.826e-08 & -9.649e-10 & 8.306e-11 & -1.520e-10 & 1.535e-08 \\
4.643e-08 & -9.649e-10 & 2.236e-04 & -6.934e-09 & -2.722e-09 & -8.960e-06 \\
-1.405e-10 & 8.306e-11 & -6.934e-09 & 7.090e-08 & -2.198e-10 & -4.618e-08 \\
-1.657e-11 & -1.520e-10 & -2.722e-09 & -2.198e-10 & 5.762e-08 & 1.440e-08 \\
2.463e-08 & 1.535e-08 & -8.960e-06 & -4.618e-08 & 1.440e-08 & 1.795e-04 \\ 
\end{tabular}
}
\caption[Normalizing an SPS one-turn map using a matched distribution]{Matrix $V^{-1} G V^{-tr}$ according to Thm. \ref{thm:m_invariance}, where $G$ is the matrix of 
Tab. \ref{tab:emit_g}. The matrix is not perfectly diagonal because the full lattice contains additional effects like non-linearities, small mismatches and numeric noise which are not considered here.}
\label{tab:emit_vgv}
\end{table}
In Tab. \ref{tab:emit_g} we show an example of a covariance matrix $G$ coming from a PyOrbit tracking simulation (without space charge) in the CERN SPS. Since $G$ is determined under the effect of small imperfections in the lattice, $V$ is not perfectly diagonalizing $G$, see Tab. \ref{tab:emit_vgv}. Let us denote by $|G - G_*|$ the error between $G$ and an ideally matched covariance matrix $G_*$ of the lattice (see Sec. \ref{subsec:emit_cov}).
We then see that the diagonal entries after diagonalizing $G$ by a symplectic matrix $S$ according to Thm. \ref{thm:williamson} belong to such a $G_*$, see Tab. \ref{tab:emit_sgs}. Remaining small errors stem from the fact that $S$ and $V$ are determined by -- and involved in -- two different procedures.\\\\
The check whether the matrix $V^{tr} S^{-1}$ is orthogonal is depicted in Tab. \ref{tab:emit_wtw}, which would be the unit matrix if $V$ would perfectly diagonalize $G$. Let us summarize this finding in the following corollary.
\begin{table}[h] \centering
\addtolength{\tabcolsep}{-0.8cm}
\hspace{-1.8cm}
\begin{tabular}{SSS}
5.793521e-08 &
7.104235e-08 &
2.001317e-04
\end{tabular}
\caption[Symplectic spectrum of a tracking covariance matrix]{Symplectic spectrum of $G$ according to Thm. \ref{thm:williamson} with respect to the horizontal, vertical and longitudinal direction. If following the procedure in Sec. \ref{subsec:emit_cov} by computing the emittances using the linear lattice optics (contained in the linear normal form $V$), the emittances are $5.793903 \cdot 10^{-8}$, $7.105457 \cdot 10^{-8}$ and $2.015400 \cdot 10^{-4}$, respectively. The agreement is remarkable.}
\label{tab:emit_sgs}
\end{table}
\begin{table}[h] \centering
\addtolength{\tabcolsep}{-0.33cm}
\hspace{-1.5cm}
\resizebox{\textwidth}{!}{
\begin{tabular}{SSSSSS}
1.002e+00 & 1.446e-03 & 4.914e-04 & -1.860e-03 & 5.269e-04 & 8.191e-05 \\
1.446e-03 & 1.006e+00 & -8.477e-05 & 5.298e-04 & -2.643e-03 & 6.747e-05 \\
4.914e-04 & -8.477e-05 & 1.117e+00 & 1.122e-04 & 7.149e-05 & -4.477e-02 \\
-1.860e-03 & 5.298e-04 & 1.122e-04 & 9.978e-01 & -1.437e-03 & -4.437e-04 \\
5.269e-04 & -2.643e-03 & 7.149e-05 & -1.437e-03 & 9.945e-01 & 7.311e-05 \\
8.191e-05 & 6.747e-05 & -4.477e-02 & -4.437e-04 & 7.311e-05 & 8.970e-01
\end{tabular}
}
\caption[Orthogonality counter-check]{Counter-check whether $W := V^{tr} S^{-1}$ is orthogonal, i.e. what is shown here is the matrix $W^{tr} W$. Hereby the symplectic $S$ was determined following the procedure in the proof of Thm. \ref{thm:williamson}. The error $|W^{tr} W - 1| = 0.168653$ can be explained by imperfections induced from the full lattice during the tracking process, resulting in divergences of $G$ towards a perfectly linearly matched solution.}
\label{tab:emit_wtw}
\end{table}
\begin{con}[Emittance from covariance matrix by symplectic diagonalization]
Let $G$ be the covariance matrix of a linearly matched particle distribution, then the emittances are given by its symplectic spectrum.
\end{con}
Moreover, by means of Thm. \ref{thm:mt} and the help of the Iwasawa decomposition, we can regain optics functions out of a covariance matrix, by utilizing this natural parameterization. The Iwasawa decomposition reads \cite{bib:terras1988}:
\begin{thm}[Iwasawa decomposition of symplectic matrices] \label{thm:iwasawa}
Let $S \in \mathrm{Sp}(2n ; \mathbb{R})$. Then there exist unique symplectic matrices $K$, $A$ and $N$ with $S = KAN$ and the following properties:
\begin{align*}
K &\in \mathrm{Sp}(2n; \mathbb{R}) \cap O(2n; \mathbb{R}) , \\
A &= \mathrm{diag}(D, D^{-1}) , \;\; D = 
\mathrm{diag}(b_1, ..., b_n) \text{ with } b_j > 0 , \\
N &= \left(\begin{array}{cc} E & F \\ 0 & E^{-tr} \end{array}\right) , \text{$E$ real unit upper triangular, $E F^{tr} = F E^{tr}$} .
\end{align*}
\end{thm}
In Thm. \ref{thm:tam} we summarize a result in \cite{bib:tinyautam2006} of how to compute such a decomposition. If we have found a symplectic diagonalization $G = S^{tr} D S$ of an invariant covariance matrix $G$, we can proceed and determine its decomposition: $S = KAN$. On the other hand, a linear normal form $V$, block-diagonalizing $M$, can also be decomposed as $V^{tr} = K'A'N'$. By Thm. \ref{thm:mt} we have $S = X V^{tr}$ with an orthosymplectic $X$, and by uniqueness of the Iwasawa decomposition it therefore follows $K = X K'$, $A = A'$ and $N = N'$. So we already found two optics factors $A'$ and $N'$ only by examination of the covariance matrix $G$.\\\\
Our next goal is to understand the nature of the remaining orthosymplectic factor $X$. As a first step note that by $V R V^{-1} = M$, $M G M^{tr} = G$ and $V D' V^{tr} := G$ we have
\begin{align*}
& V R V^{-1} G V^{-tr} R^{tr} V^{tr} = M G M^{tr} = G , \\
\Rightarrow \;\; & R V^{-1} S^{tr} D S V^{-tr} R^{tr} = V^{-1} G V^{-tr} = D' ,
\end{align*}
and since $R$ commutes with $D'$ it follows
\begin{equation}
V^{-1} S^{tr} D S V^{-tr} = D' . \label{eq:xdx}
\end{equation}
For the symplectic spectrum we could have used any (positive definite) covariance matrix, but for a determination of the optics functions we will now have to make one additional assumption: Namely that the symplectic spectrum (respectively emittances) is not degenerate, which means that all emittances are mutually distinguishable. By suitable orthosymplectic permutations (their construction is given in Prop. \ref{dfnsz:strns}) on $K$ and $K'$ let us arrange them so that without loss of generality $D' = D$ and the $2n$ entries of $D$ are in the following order:
\begin{equation}
d_1 = d_{n + 1} < d_2 = d_{n + 2} < ... < d_n = d_{2n} .
\end{equation}
It follows from Eq. \eqref{eq:xdx}, which now reads $D X = X D$, that $d_i X_{ij} = X_{ij} d_j$ or, on other words, if $d_i \not = d_j$ then $X_{ij} = 0$, and, by exchanging symbols, also $X_{ji} = 0$ in that case. If we now take a look at the indices $k$ and $k + n$, then we see that for all $l \not \in \{ k, k + n \}$ it holds $X_{k, l} = 0 = X_{l, k}$ and $X_{l, k + n} = 0 = X_{k + n, l}$, so concerning these rows and columns $X$ must have the following form:
\begin{align*}
X &= \left(\begin{array}{ccccccccccc} 
& & & 0 & & & & 0 & & & \\
& & & \vdots & & & & \vdots & & & \\
& & & 0 & & & & 0 & & & \\
0 & \cdots & 0 & * & 0 & \cdots & 0 & * & 0 & \cdots & 0 \\
 &  &  & 0 & & & & 0 & & & \\
 &  &  & \vdots & & & & \vdots & & & \\
 &  &  & 0 & & & & 0 & & & \\
0 & \cdots & 0 & * & 0 & \cdots & 0 & * & 0 & \cdots & 0 \\
& & & 0 & & & & 0 & & & \\
& & & \vdots & & & & \vdots & & & \\
& & & 0 & & & & 0 & & & 
\end{array}\right)
\begin{array}{l}
\hphantom{0} \\
\hphantom{\vdots} \\
\hphantom{0} \\
\gets k \\
\hphantom{0} \\
\hphantom{\vdots} \vspace{0.22cm} \\
\hphantom{0} \\
\gets k + n \\
\hphantom{0}\\
\hphantom{\vdots} \\
\hphantom{0}
\end{array} \\
\hphantom{X} & \hphantom{=} \hspace{-0.68cm} 
\left.\begin{array}{cccccccc} 
\hphantom{0} & \hphantom{\cdots} & \hphantom{0} & \hspace{0.78cm} k \uparrow & \hspace{1.4cm} \uparrow k + n
\end{array}\right.
\end{align*}
By our assumption on the mutually distinguishable pairs, and because every such group of four entries must be orthosymplectic on its own, we must have $X = D_1 + J D_2$ as in lemma \ref{lem:r_invariance} with $D_1^2 + D_2^2 = 1$. By Rmk. \ref{rmk:freedom_v} this was the freedom in the choice of $V$, which leaves the individual Floquet-planes invariant. So without loss of generally $K$ can be considered as equal to $K'$ up to a symplectic permutation matrix which is exchanging the order of the symplectic spectrum and up to this freedom. Let us summarize this result:
\begin{con} \label{con:iwa_opt}
If $G$ is the covariance matrix of a linearly matched particle distribution, and $S^{tr} G S = D$ a symplectic diagonalization with mutually distinguishable emittances, then the optics functions are contained in the terms $K$, $A$ and $N$ of the Iwasawa decomposition $S = KAN$. In particular we can regain the linear map to Floquet-space from $S$, block-diagonalizing $M$. Hereby, the factor $K$ may contain a freedom of the form $D_1 + J D_2$ with $D_1^2 + D_2^2 = 1$, i.e. of the form $SO(2)^n$, and a suitable symplectic permutation of the components.\footnote{In other words: The freedom is a block-rotation as in Thm. \ref{thm:s_diag_r}.}
\end{con}
In analogy to what happens in the 2D case (see Example \ref{xmpl:2d} below), the term $N$ can be understood as a \textit{lens}- or \textit{drift factor} and the term $A$ can be described as a \textit{magnification} or \textit{squeezing factor}. As we shall see in Example \ref{xmpl:4d_cpl}, the term $K$ may also contain coupling and optic terms in a non-trivial fashion. Under the condition in conclusion \ref{con:iwa_opt}, all those three factors, and therefore the linear parameterization to normal form, are thus dependent on the $n(2n + 1)$ parameters of the covariance matrix $G$ only, which are 3, 10 and 21 in the 2, 4 and 6 dimensional cases respectively. The freedom in the coupling term reduces these numbers by $n$, so we can expect to have 2, 8 and 18 independent optics parameters in these cases.\\\\
This freedom represents our inability to extract the tunes out of the covariance matrix alone and so we can expect that in repetitive measurements the additional $SO(2)^n$-freedom enters into the covariance matrices statistically. As we will see in Example \ref{xmpl:4d_cpl}, $K$ is not in general of the form $D_1 + J D_2$. In cases one wants to obtain the exact coupling terms of an underlying model, this will require a careful analysis in order to disentangle the freedom from these coupling terms. In any case, if the emittances of the given covariance matrix are mutually distinguishable, we obtain a linear normal form map from it which block-diagonalizes $M$ (with a certain error), as demonstrated in Tabs. \ref{tab:block_m} and \ref{tab:block_m_via_g}.\\\\
\begin{table}[h] \centering
\hspace{-1.5cm}
\resizebox{0.82\textwidth}{!}{
\begin{tabular}{SSSSSS}
0.619812 & 0.0 & 0.0 & 0.784750 & 0.0 & 0.0 \\
0.0 & -0.149110 & 0.0 & 0.0 & 0.988821 & 0.0 \\
0.0 & 0.0 & 0.995494 & 0.0 & 0.0 & -0.094829 \\
-0.784750 & 0.0 & 0.0 & 0.619812 & 0.0 & 0.0 \\
0.0 & -0.988821 & 0.0 & 0.0 & -0.149110 & 0.0 \\
0.0 & 0.0 & 0.094829 & 0.0 & 0.0 & 0.995494
\end{tabular}
}
\caption[Block-diagonal SPS one-turn map by normal form]{Block-diagonal one-turn map $M$ of the SPS at the tune $Q_x = 20.144$ for the same optics as in Tab.
\ref{tab:emit_vgv}. Note that in this example the entries are ordered with respect to $x, y, z, p_x, p_y, p_z$.}
\label{tab:block_m}
\end{table}
\begin{table}[h] \centering
\hspace{-1.5cm}
\resizebox{0.82\textwidth}{!}{
\begin{tabular}{SSSSSS}
0.617887 & -0.000045 & -0.000109 & 0.783535 & -0.008843 & -0.000242 \\
0.003670 & -0.145196 & 0.000033 & 0.006804 & 0.993562 & -0.000162 \\
0.000208 & 0.000160 & 0.996339 & -0.000144 & -0.000019 & -0.084259 \\
-0.786018 & 0.007309 & -0.000117 & 0.621575 & 0.002528 & -0.000006 \\
-0.005988 & -0.984076 & 0.000010 & 0.006479 & -0.152863 & 0.000042 \\
-0.000136 & 0.000050 & 0.106734 & -0.000040 & 0.000014 & 0.994648
\end{tabular}
}
\caption[Block-diagonal SPS one-turn map by using a covariance matrix]{Block-diagonal form of $M$ now reconstructed entirely on the information of the particle distribution of Tab. \ref{tab:emit_g}. I.e. what is shown here is $S^{-tr} M S^{tr}$ where $S$ was extracted from the covariance matrix using Thm. \ref{thm:williamson}. The agreement to Tab. \ref{tab:block_m} is very good. Note that we had to conjugate this matrix by the symplectic transposition $T_{12}$ which are discussed in Prop. \ref{dfnsz:strns}.}
\label{tab:block_m_via_g}
\end{table}
As we shall see in examples \ref{xmpl:2d} and \ref{xmpl:4d_cpl}, the three diagonal terms of $A$ are directly related to the three classical optics beta functions of the Courant-Snyder parameterization in the uncoupled case. But as they are also appearing in a general coupled situation and are determined by a natural procedure -- and also to distinguish them from alternative 'generalized' $\beta$-functions discussed below -- we call these 'natural' beta functions. In Tab. \ref{tab:iwasawa_a} we show the matrices $A$ determined from the normal form map $V^{tr}$ and the symplectic map $S$ (coming from the covariance matrix $G$ of our SPS example) and check their relation to a \texttt{twiss} output from MAD-X.\\\\
\begin{table}[h] \centering
\hspace{-1.5cm}
\resizebox{0.82\textwidth}{!}{
\begin{tabular}{SSSSSS} \toprule
9.187210 & 0.0 & 0.0 & 0.0 & 0.0 & 0.0 \\
0.0 & 6.338581 & 0.0 & 0.0 & 0.0 & 0.0 \\
0.0 & 0.0 & 11.389408 & 0.0 & 0.0 & 0.0 \\
0.0 & 0.0 & 0.0 & 0.108847 & 0.0 & 0.0 \\
0.0 & 0.0 & 0.0 & 0.0 & 0.157764 & 0.0 \\
0.0 & 0.0 & 0.0 & 0.0 & 0.0 & 0.087801 \\ \bottomrule
&&&&& \\ \toprule
9.192319 & 0.0 & 0.0 & 0.0 & 0.0 & 0.0 \\
0.0 & 6.358151 & 0.0 & 0.0 & 0.0 & 0.0 \\
0.0 & 0.0 & 12.082668 & 0.0 & 0.0 & 0.0 \\
0.0 & 0.0 & 0.0 & 0.108786 & 0.0 & 0.0 \\
0.0 & 0.0 & 0.0 & 0.0 & 0.157278 & 0.0 \\
0.0 & 0.0 & 0.0 & 0.0 & 0.0 & 0.082763 \\ \bottomrule
\end{tabular}
}
\caption[Iwasawa factors of an SPS tracking example]{Iwasawa factors $A$ determined from either the linear normal form map $V^{tr}$ (top) or the symplectic map $S$ coming from the covariance matrix $G$ (bottom) of our SPS tracking example of Tab. \ref{tab:emit_g}. Small differences occur because the covariance matrix $G$ comes from a real tracking example with small lattice non-linearities, mismatches and numeric noise. From a MAD-X \texttt{twiss} command we compare the 1-1 and 2-2 entries (exemplary for the bottom matrix) and obtain good agreement: $\texttt{BETX} = 84.404930$, $9.192319^2 = 84.498729$ and $\texttt{BETY} = 40.178387$, $6.358151^2 = 40.426084$.}
\label{tab:iwasawa_a}
\end{table}
Independently on the problem of finding a faithful parameterization for the optics, one can also introduce 'generalized' $\beta$-functions, 
as proposed in Ref. \cite{bib:wolski2006}. The main idea is to use the analogy of the two-dimensional case, in which the $\beta$-function 
appears as coefficient in the well-known relation between
the emittance and the rms beam size, e.g. $\langle x^2 \rangle \propto \beta_x \epsilon_x$. Since the optics functions $K$, $A$ and $N$
are always the same for beams with non-degenerated spectrum, one can use the symplectic $S$ (or $V^{tr}$ of the linear normal form) in order to obtain such relations between the second
moments of the matched distribution and the emittances:
\begin{equation}
\langle x_i x_j \rangle = e_i^{tr} G e_j = e_i^{tr} S^{tr} D S e_j = \sum_{l = 1}^n (S_{il} S_{jl} + S_{i, n + l} S_{j, n + l}) \Lambda_l =: \sum_{l = 1}^n \beta_{ij}^{\hspace{.25cm} l} \Lambda_l . \label{eq:beta_tensor}
\end{equation}
Algebraically speaking they correspond to the coefficients of the parameterization $\mathbb{R}^n \hookrightarrow \mathrm{Sym}(\mathbb{R}^{2n}) \subset \mathbb{R}^{2n} \otimes \mathbb{R}^{2n}$ of the $M$-invariant covariance matrices and are by definition related to the $Z_l$ in Thm. \ref{cor:real_sym_inv_basis} via $2 \beta_{ij}^{\hspace{.25cm} l} = (Z_l)_{ij}$. From the analogy to the 2D case (cf. $G$ in example \ref{xmpl:2d}), in which the situation goes over without coupling, one can identify $\beta_x$ and $\beta_y$ with $\beta_{11}^{\hspace{.25cm} 1}$ and $\beta_{22}^{\hspace{.25cm} 2}$, $\alpha_x$ and $\alpha_y$ with $- \beta_{14}^{\hspace{.25cm} 1}$ and $- \beta_{25}^{\hspace{.25cm} 2}$ and $\gamma_x$ and $\gamma_y$ with $\beta_{44}^{\hspace{.25cm} 1}$ and $\beta_{55}^{\hspace{.25cm} 2}$ respectively.\\\\
As was shown in Refs. \cite{bib:wolski2006, bib:ohkawa2007, bib:alexahin2013}, these coefficients have the feature that one can also find expressions for the dispersions. The idea is that in the
classical 2D theory the betatron motion and the dispersive part are uncorrelated: Let $n = 3$ and $\eta_k$ for $k = 1, 2, 4, 5$ denote the dispersion function with respect to direction $k$, where the 6th component $x_6$ corresponds to the energy offset $\delta p/p_0$. Then from $x_k = x_{k, \beta} + \eta_k x_6$ it follows 
$\langle x_k x_6 \rangle = \eta_k \langle x_6^2 \rangle$, which translates to 
\begin{equation}
\sum_{l = 1}^3 \beta_{k, 6}^{\hspace{.35cm} l} \Lambda_l = \eta_k \sum_{l = 1}^3 \beta_{6, 6}^{\hspace{.35cm} l} \Lambda_l \;\; \Rightarrow \;\; \eta_k \cong \frac{\beta_{k, 6}^{\hspace{.35cm} 3}}{\beta_{6, 6}^{\hspace{.35cm} 3}} , \label{eq:disp_from_beta}
\end{equation}
hereby the last equation was assuming a decoupled case, similar to what was previously done to relate the $\beta_{ij}^{\hspace{.25cm} l}$ to the classical optics functions. In Tab. \ref{tab:beta_tensor} we show the general $\beta$-functions
for our SPS example and their agreement with those coming from a MAD-X \texttt{twiss} command. It should therefore be clear how useful these statistical definitions are when it comes to situations where one can not easily determine effective optics functions otherwise, for example in the scenario of a PIC code with space charge.\\\\
\begin{table}[h] \centering
\hspace{-1.5cm}
\resizebox{0.82\textwidth}{!}{
\begin{tabular}{SSSSSS} \toprule
 84.476949 &   0.774801 &  -0.882595 &   1.626052 &  -0.023067 &  -0.000297 \\
  0.774801 &   0.010308 &  -0.010263 &   0.021069 &  -0.000164 &  -0.000006 \\
 -0.882595 &  -0.010263 &   0.010690 &  -0.021157 &   0.000209 &   0.000005 \\
  1.626052 &   0.021069 &  -0.021157 &   0.043131 &  -0.000352 &  -0.000011 \\
 -0.023067 &  -0.000164 &   0.000209 &  -0.000352 &   0.000007 &   0.000000 \\
 -0.000297 &  -0.000006 &   0.000005 &  -0.000011 &   0.000000 &   0.000000 \\
&&&&&\\
  0.020908 &  -0.757982 &  -0.000262 &   0.000344 &   0.028889 &  -0.000000 \\
 -0.757982 &  40.415782 &   0.006499 &  -0.017789 &  -0.854362 &   0.000014 \\
 -0.000262 &   0.006499 &   0.000004 &  -0.000003 &  -0.000406 &   0.000000 \\
  0.000344 &  -0.017789 &  -0.000003 &   0.000008 &   0.000397 &  -0.000000 \\
  0.028889 &  -0.854362 &  -0.000406 &   0.000397 &   0.042793 &  -0.000000 \\
 -0.000000 &   0.000014 &   0.000000 &  -0.000000 &  -0.000000 &   0.000000 \\
&&&&&\\
  0.000880 &   0.000021 &   0.066726 &  -0.000008 &  -0.000001 &  -0.002408 \\
  0.000021 &   0.000001 &   0.000440 &  -0.000000 &  -0.000000 &  -0.000059 \\
  0.066726 &   0.000440 & 145.988045 &   0.000864 &   0.000027 &   0.008947 \\
 -0.000008 &  -0.000000 &   0.000864 &   0.000000 &   0.000000 &   0.000025 \\
 -0.000001 &  -0.000000 &   0.000027 &   0.000000 &   0.000000 &   0.000002 \\
 -0.002408 &  -0.000059 &   0.008947 &   0.000025 &   0.000002 &   0.006850 \\ \bottomrule
\end{tabular}
}
\caption[Generalized $\beta$-functions]{Components of the $\beta$-tensor of Eq. \eqref{eq:beta_tensor} for our SPS example (at the working point $Q_x = 20.144$). The following entries are going over into the classical linear lattice parameters in case of no coupling, and are given as follows in form of the output of a MAD-X \texttt{twiss} command: 111: $\texttt{BETX} = 84.404930$, 222: $\texttt{BETY} = 40.178387$, 141: $\texttt{ALFX} = -1.627366$, 252: $\texttt{ALFY} = 0.850149$. $\gamma$-check: 441: $\gamma_x = 0.043224$, 552: $\gamma_y = 0.042878$. Dispersion: $\texttt{DX} = -0.343296$, 163/663: $-0.351470$, $\texttt{DY} = -0.004177$, 263/663: $-0.008667$. $\texttt{DPX} = 0.003525$, 463/663: $0.003678$. $\texttt{DPY} = 0.000319$, 563/663: $0.000247$. Note that the signs of the entries to the $\alpha$'s are reversed, as expected.
}
\label{tab:beta_tensor}
\end{table}
Having the dispersion parameters at hand, which were determined by including assumptions of the origin of the one-turn map $M$, one can attempt to recover the tunes of the unknown optics -- in principle. However, as we shall see, the sensitivity with respect to the dispersion terms is very high. This indicates that such an undertaking, by purely examining covariance matrices, might require more elaborate methods (and probably also better statistics by including more particles):\\\\
For the next considerations we change to the ordering $x, p_x, y, p_y, z, p_z$ and write the one-turn-map $M$ into 4 and 2-blocks
\begin{equation}
M = \left(\begin{array}{cc} M_4 & A \\ B & C \end{array}\right) .
\end{equation}
Let $X^{tr} = (X_1, 0, p_z)$ for $X_1 \in \mathbb{R}^4$ describe a 4D closed-orbit solution of $M$, i.e.
$M_4 X_1 + A (0, p_z)^{tr} = X_1$, which translates to $X_1 = (1 - M_4)^{-1} A (0, p_z)^{tr}$ and so
$\mathcal{D} := (1 - M_4)^{-1} A$ contains the 4 known dispersion terms \texttt{DX}, \texttt{DPX}, \texttt{DY} and \texttt{DPY} in the second column. Although the terms in all columns of $\mathcal{D}$ can be approximated in an analog fashion as the left-hand side of Eq. \eqref{eq:disp_from_beta} if correlations between $z$ and the other spatial coordinates are small, and by taking into account only correlations between $z$ and $p_z$, let us assume that the first column is unknown. The relation $V^{-1} M V = R$ ($V$ any symplectic map block-diagonalizing $M$, in particular having in mind a map coming from a covariance matrix) reads in this context
\begin{align}
\left(\begin{array}{cc} M_4 & A \\ B & C \end{array} \right) 
\left(\begin{array}{cc} V_{11} & V_{12} \\ V_{21} & V_{22} \end{array}\right)
&= \left(\begin{array}{cc} V_{11} & V_{12} \\ V_{21} & V_{22} \end{array}\right)
\left(\begin{array}{cc} R_4 & 0 \\ 0 & R_2 \end{array}\right) \nonumber \\
\Leftrightarrow \;\; \left(\begin{array}{cc} 
M_4 V_{11} + A V_{21} & M_4 V_{12} + A V_{22} \\
B V_{11} + C V_{21} & B V_{12} + C V_{22} \end{array}\right) &= 
\left(\begin{array}{cc} V_{11} R_4 & V_{12} R_2 \\ V_{21} R_4 & V_{22} R_2 \end{array}\right) ,
\end{align}
hereby $R_4$ is a $4 \times 4$ block-diagonal rotation matrix, as discussed in Thm. \ref{thm:s_diag_r}. We can now use the first row to express the unknown $M_4$ by $\mathcal{D}$ and the rotation matrices $R_2$ and $R_4$ which contain the unknown tunes:
\begin{subequations}
\begin{align}
M_4 (V_{11} - \mathcal{D} V_{21}) + \mathcal{D} V_{21} = M_4 V_{11} + (1 - M_4) \mathcal{D} V_{21} &= V_{11} R_4 , \\
M_4 (V_{12} - \mathcal{D} V_{22}) + \mathcal{D} V_{22} = M_4 V_{12} + (1 - M_4) \mathcal{D} V_{22} &= V_{12} R_2 .
\end{align}
\end{subequations}
If we assume that the $4 \times 4$-map $V_{11} - \mathcal{D} V_{21}$ is invertible (which is the case in our SPS example), we can eliminate $M_4$ to obtain
\begin{equation}
(V_{11} R_4 - \mathcal{D} V_{21})(V_{11} - \mathcal{D} V_{21})^{-1} (V_{12} - \mathcal{D} V_{22}) + 
\mathcal{D} V_{22} - V_{12} R_2 = 0 .
\label{eq:tune_from_sigma}
\end{equation}
This corresponds to a system of 8 equations for 7 unknown parameters (3 tunes and the 4 entries of the first column of $\mathcal{D}$). For the given data of our SPS example it turned out, however, that the sensitivity of this problem on the dispersion terms is too high, even if assuming a known $z$-tune, as is summarized in Tab. \ref{tab:tune_sens}.\\\\
We conclude this section with three examples.
\begin{table}
\centering
\addtolength{\tabcolsep}{-0.4cm}
\hspace{-0.4cm}
\begin{tabular}{SSSS} \toprule
{${\mathcal{D}_{j1}}$ \text{inexact}}  & {${\mathcal{D}_{j2}}$ \text{inexact}} & 
{${\mathcal{D}_{j1}}$ \text{exact}}  & {${\mathcal{D}_{j2}}$ \text{exact}} \\ \hline
6.235119e-04 & -3.400830e-01 & 1.143900e-03 & -3.440703e-01 \\
8.110447e-06 & 3.533734e-03 & 2.491066e-05 &  3.518960e-03 \\
8.043134e-07 & -4.242339e-03 & 7.223080e-06 & -4.235338e-03 \\
-1.321050e-07 &  3.187020e-04 & -5.912169e-08 & 3.195229e-04 \\ \bottomrule
\end{tabular}
\addtolength{\tabcolsep}{0.8cm} \\ \vspace{0.2cm}
\begin{tabular}{l|c|c} 
 &  ${\mathcal{D}_{j2}}$ {\text{inexact}} & ${\mathcal{D}_{j2}}$ {\text{exact}} \\ \hline
${\mathcal{D}_{j1}}$ {\text{inexact}} & \text{False} & \text{True} \\ \hline
${\mathcal{D}_{j1}}$ {\text{exact}}  & \text{False} & \text{True} 
\end{tabular}
\caption[Convergence tests for tune determination]{Convergence successes of the Nelder-Mead simplex algorithm \cite{bib:nelder1965, bib:gao2012} applied to Eq. \eqref{eq:tune_from_sigma}, assuming a tune $Q_z = -1.511528 \cdot 10^{-2}$ of our SPS example, in dependency of the four combinations of exact and inexact initial parameter columns. We have used 6 iteration restarts of the optimization routine. The inexact values are determined from the tracking example, in which the values of the first column were estimated by using the assumption that the $z$-motion is uncorrelated to the first four coordinates. As can be seen, the success depends on the precision of the dispersion terms which are in the second column of $\mathcal{D}$.}
\label{tab:tune_sens}
\end{table}

\begin{xmpl}[2D] \label{xmpl:2d} For
  $\alpha, \varphi \in \mathbb{R}$ and $\mathbb{R} \ni \beta, \gamma > 0$ with
  $\beta \gamma = 1 + \alpha^2$ consider a linear transport map $M$, see Ref. \cite{bib:lee2012}, and a positive definite symmetric $G$
  \begin{subequations}
    \begin{align}
      M &:= \left(
          \begin{array}{cc}
            \cos(\varphi) + \alpha \sin(\varphi) & \beta \sin(\varphi) \\
            - \gamma \sin(\varphi) & \cos(\varphi) - \alpha \sin(\varphi)
          \end{array}\right) , \label{eq:xmpl_2dm} \\
      G &:= \left(
          \begin{array}{cc}
            \beta & -\alpha \\
            -\alpha & \gamma
          \end{array}\right) . \label{eq:xmpl_2dg}
    \end{align}
  \end{subequations}
  Then one can show that $M$ is symplectic and it holds
  $M G M^{tr} = G$. Furthermore, the symplectic matrices
  \begin{subequations}
    \begin{align}
      S_1 &:= \left(
            \begin{array}{cc}
              \sqrt{\beta} & -\alpha / \sqrt{\beta} \\
              0 & 1/\sqrt{\beta}
            \end{array}
                  \right) , \label{eq:xmpl_2ds1} \\
      S_2 &:= \left(
            \begin{array}{cc}
              1/\sqrt{\gamma} & 0 \\
              -\alpha / \sqrt{\gamma} & \sqrt{\gamma}
            \end{array}
                                      \right) , \label{eq:xmpl_2ds2}
    \end{align}
  \end{subequations}
  satisfy $G = S_i^{tr} S_i$ (in particular, $G$ is positive definite). This means that we
  have all requirements of Thm. \ref{thm:mt} and $S_2 S_1^{-1}$ must
  be orthogonal. Indeed we have the following Iwasawa decompositions
  \begin{subequations}
    \begin{align}
      S_1 &= A_1 N_1 = \left(
            \begin{array}{cc}
              \sqrt{\beta} & 0 \\
              0 & 1/\sqrt{\beta}
            \end{array} 
                  \right)
                  \left(
                  \begin{array}{cc}
                    1 & -\alpha / \beta \\
                    0 & 1
                  \end{array}
                        \right) , \\
      S_2 &= K_2 A_2 N_2 = \frac{1}{\sqrt{\beta \gamma}} \left(
            \begin{array}{cc}
              1 & \alpha \\
              -\alpha & 1
            \end{array}
                       \right)
                       \left(
                       \begin{array}{cc}
                         \sqrt{\beta} & 0 \\
                         0 & 1/\sqrt{\beta}
                       \end{array}
                             \right)
                             \left(
                             \begin{array}{cc}
                               1 & -\alpha / \beta \\
                               0 & 1
                             \end{array}
                                   \right) ,
    \end{align}
  \end{subequations}
  and we see $A_1 = A_2$ and $N_1 = N_2$ and $S_2 S_1^{-1} = K_2$ is
  orthogonal, as claimed by Thm. \ref{thm:mt}. The two independent optics parameters $\alpha$
and $\beta$ can be regained by comparison of the above Iwasawa factors with the ones obtained by any matched particle distribution according to conclusion \ref{con:iwa_opt}. Since $M G M^{tr} = G$ holds, we have $M^{-tr} G^{-1} M^{-1} = G^{-1}$ and therefore $G^{-1} = M^{tr} G^{-1} M$, which means that the associated quadratic form $g(z) := z^{tr} G^{-1} z$ is $M$-invariant. We have $g(z) = z^{tr} N^{-1} A^{-1} A^{-tr} N^{-tr} z = ((AN)^{-tr})^* 1_2(z)$ or, reversely, $(A^{tr})^* (N^{tr})^* g = 1_2$. The effects of these two operations are illustrated in Fig. \ref{fig:2dcs} and appear frequently in elementary particle accelerator textbooks.
\end{xmpl}

\begin{figure}[h]
\begin{minipage}[t]{\textwidth}
\centering
\begin{tikzpicture}[>=latex',line join=bevel]

\pgftransformcm{1.0}{0.0}{0.9}{1.0}{\pgfpoint{0}{0}} 
\draw (-5, 0) ellipse (0.667cm and 1.5cm); 
\pgftransformreset

\def\arrow{
  (10.75:1.1) -- (6.5:1) arc (6.25:120:1) [rounded corners=0.5] --
  (120:0.9) [rounded corners=1] -- (130:1.1) [rounded corners=0.5] --
  (120:1.3) [sharp corners] -- (120:1.2) arc (120:5.25:1.2)
  [rounded corners=1] -- (10.75:1.1) -- (6.5:1) -- cycle
}

\definecolor{darkblue}{rgb}{0.2,0.2,0.6}

\tikzset{
  ashadow/.style={opacity=.25, shadow xshift=0.07, shadow yshift=-0.07},
}

\pgftransformcm{-1.2}{0.0}{0.0}{1.2}{\pgfpoint{0}{0}}
    \draw[color=darkblue, bottom color=white!90!black, top color=white!30, %
    drop shadow={ashadow, color=blue!60!black}] [xshift=1.6cm, yshift=0.4cm, rotate=20] \arrow;
\pgftransformreset

\pgftransformcm{-1.2}{0.0}{0.0}{1.2}{\pgfpoint{0}{0}}
    \draw[color=darkblue, bottom color=white!90!black, top color=white!30, %
    drop shadow={ashadow, color=blue!60!black}] [xshift=-2.2cm, yshift=0.4cm, rotate=20] \arrow;
\pgftransformreset

\draw (-1.8, 2.6) node[below] {$(N^{tr})^*$};
\draw (2.6, 2.5) node[below] {$A^*$};

\draw[->] (-7.0, 0) -- (-3.0, 0) node[right]{$x$};
\draw[->] (-5.0, -2.0) -- (-5.0, 2.0) node[above]{$p_x$};

\draw (0, 0) ellipse (0.667cm and 1.5cm);
\draw[->] (-2.0, 0) -- (2.0, 0) node[right]{$x$};
\draw[->] (0, -2.0) -- (0, 2.0) node[above]{$p_x$};

\draw (5, 0) ellipse (1cm and 1cm);
\draw[->] (3.0, 0) -- (7.0, 0) node[right]{$x$};
\draw[->] (5.0, -2.0) -- (5.0, 2.0) node[above]{$p_x$};

\end{tikzpicture}

\end{minipage}
\caption[Courant-Snyder decomposition]{Effect of the two operations $(N^{tr})^*$ and $A^*$ on a phase space ellipse given by the quadratic form $g$ in Example \ref{xmpl:2d}, using $\alpha > 0$.}
\label{fig:2dcs}
\end{figure}
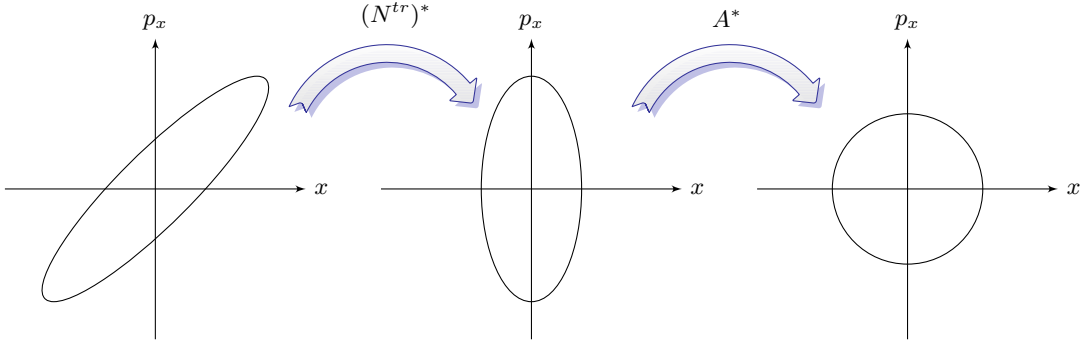

\begin{xmpl}[4D] \label{xmpl:4d}
Let $G = \left(\begin{array}{cc} A & C \\ C^{tr} & B \end{array}\right)$ be the covariance matrix of a linearly matched particle distribution in a 4D tracking routine. Then its symplectic spectrum, which are the two emittances, are given by
\begin{equation}
\epsilon_{1, 2} = \frac{1}{\sqrt{2}} \sqrt{\Delta G \pm \sqrt{\Delta G^2 - 4 \det(G)} } , \label{eq:4dcor}
\end{equation}
where $\Delta G := \det(A) + \det(B) + 2 \det(C)$. This formula for the spectrum of a 4D covariance matrix can be found in Ref. \cite{bib:pirandola2009} in a different context. From Eq. \eqref{eq:4dcor} we regain two familiar symplectic invariants:
\begin{subequations}
\begin{eqnarray}
\det (G) &=& \epsilon_1^2 \epsilon_2^2 , \\
\Delta G &=& \epsilon_1^2 + \epsilon_2^2 .
\end{eqnarray}
The first one appears for example in Ref. \cite{bib:lebedev2010}.
\end{subequations}
\end{xmpl}

For the next example we need some more machinery.

\begin{lem} \label{lem:ldl_2d} Let
$A \in \mathbb{K}^{2 \times 2}$ 
be symmetric with $g > 0$ so
that
  \[
  A = \left(
    \begin{array}{cc}
      g & a \\
      a & b
    \end{array}\right) .
  \]
  Then the LDL-Cholesky factorization $A = Q^{tr} H Q$ of $A$ is given by
  \begin{align*}
    H &= \mathrm{diag}(g, \mathrm{det}(A)/g) , \\
    Q &= \left(
        \begin{array}{cc}
          1 & a/g \\ 0 & 1 
        \end{array}
                         \right) .
  \end{align*}
\end{lem}
\begin{proof}
  The proof is left to the reader. Note that this equation holds also for $g < 0$.
\end{proof}

\begin{thm}[T. Y. Tam] \label{thm:tam} Let
  $\mathbb{K} \in \{\mathbb{R}, \mathbb{C}\}$ and
  $X \in \mathrm{Sp}(2n; \mathbb{K})$ with
  \[
  X^H X =: \left(
    \begin{array}{cc}
      A_1 & B_1 \\
      B_1^H & D_1 
    \end{array}
  \right) .
  \]
  Let $A_1 = Q^H H Q$ be the LDL-Cholesky factorization of the
  positive definite $A_1$, where $Q$ is unit upper triangular and $H$
  positive diagonal. Then the Iwasawa factors $A$ and $N$ of $X = KAN$
  can be computed by
  \begin{subequations}
    \begin{align}
      A &= \left(
          \begin{array}{cc}
            H^{1/2} & 0 \\
            0 & H^{-1/2}
          \end{array}
                \right) , \\
      N &= \left(
          \begin{array}{cc}
            Q & Q A_1^{-1} B_1 \\
            0 & Q^{-tr}
          \end{array}
                \right) =
                \left(
                \begin{array}{cc}
                  Q & H^{-1} Q^{-tr} B_1 \\
                  0 & Q^{-tr}
                \end{array}
                      \right) .
    \end{align} 
  \end{subequations}
\end{thm}
\begin{proof}
  See Ref. \cite{bib:tinyautam2006}.
\end{proof}

\begin{xmpl}[4D optics with a single coupling term] \label{xmpl:4d_cpl}
  Similar as in Example \ref{xmpl:2d} consider for $i = 1, 2$:
  $\alpha_i, \varphi_i \in \mathbb{R}$ and $\mathbb{R} \ni \beta_i, \gamma_i > 0$ with
  $\beta_i \gamma_i = 1 + \alpha_i^2$ and
  $\varphi_1 \not = \varphi_2$ the matrices
  $M_1, M_2, G_1, G_2 \in \mathbb{R}^{2 \times 2}$.
  Let $\diamond$ be the operator defined in Prop. \ref{dfnsz:oplus}. From the properties we have
$J_4 = J_2 \diamond J_2$. Let
  \begin{align*}
    \tilde M &:= M_1 \diamond M_2 , \\
    \tilde G &:= G_1 \diamond G_2 .
  \end{align*}
  It follows from Prop. \ref{dfnsz:oplus} with the
  symplecticity of $M_i$ that $\tilde M$ is symplectic and furthermore
  that $\tilde G$ is
  symmetric and $\tilde M \tilde G \tilde M^{tr} = \tilde G$. From
  Prop. \ref{dfnsz:oplus} it also follows that $\tilde M$ has
  four complex eigenvalues of the form $\exp(\pm i \varphi_j)$ and by
  assumption it is guaranteed that no eigenvalue equals $\pm 1$.
  Moreover it follows that $\tilde G$ is positive definite. 
$\tilde M$ can be interpreted as an uncoupled lattice. Let us now introduce
a basic coupling term; for
  $\psi \in \mathbb{R}$ set $c := \cos(\psi)$, $s := \sin(\psi)$ and
  $S^{(2)}_{ij} \in \mathbb{R}^{2 \times 2}$ as in Example
  \ref{xmpl:2d}, Eqs. \eqref{eq:xmpl_2ds1}, \eqref{eq:xmpl_2ds2}, with
  $S^{(2)}_{ij}$ using $\alpha_j, \beta_j$ and $\gamma_j$.
  \begin{align*}
    U &:= \left(
        \begin{array}{cc}
          c & -s \\ s & c
        \end{array}
                        \right) \in \mathbb{R}^{2 \times 2}, \\
    V &:= \mathrm{diag}(U, U) \in \mathbb{R}^{4 \times 4}, \\
    \tilde S_{ij} &:= S^{(2)}_{i1} \diamond S^{(2)}_{j2} .
  \end{align*}
  By the properties of the operator $\diamond$, all four $\tilde S_{ij}$
  are symplectic.  $V$ is clearly orthogonal and also symplectic:
  \[
  V^{tr} J_4 V = \mathrm{diag}(U^{-1}, U^{-1}) \left(
    \begin{array}{cc}
      0 & U \\
      - U & 0
    \end{array}
  \right) = J_4 .
  \]
  The symplectic $M := V^{tr} \tilde M V$ satisfies $M G M^{tr} = G$
  with symmetric and positive definite $G := V^{tr} \tilde G V$:
  \[
  M G M^{tr} = V^{tr} \tilde M V V^{tr} \tilde G V V^{tr} \tilde M^{tr} V   = V^{tr} \tilde G V = G .
  \]
  The interpretation of $G$ is that it describes the covariance matrix of a linearly matched particle distribution in this coupled optics. From Example \ref{xmpl:2d} we know that
  $\tilde S_{ij}^{tr} \tilde S_{ij} = \tilde G$ must hold. Therefore
  the symplectic $S_{ij} := \tilde S_{ij} V$ are diagonalizing $G$ as
  in Williamson's Theorem:
  $S_{ij}^{tr} S_{ij} = V^{tr} \tilde S_{ij}^{tr} \tilde S_{ij} V =
  V^{tr} \tilde G V = G$.
  Let us now use lemma \ref{lem:ldl_2d} and Thm. \ref{thm:tam} to
  compute the Iwasawa decomposition exemplary in the case of $S_{11}$. By construction:
\begin{equation}
\tilde S_{11} = \left(
\begin{array}{cccc}
\sqrt{\beta_1} & 0 & -\alpha_1/\sqrt{\beta_1} & 0 \\
0 & \sqrt{\beta_2} & 0 & -\alpha_2/\sqrt{\beta_2} \\
0 & 0 & 1/\sqrt{\beta_1} & 0 \\
0 & 0 & 0 & 1/\sqrt{\beta_2}
\end{array}
\right) =: 
\left(
\begin{array}{cc}
s_{11} & s_{12} \\
0 & s_{22}
\end{array}
\right) ,
\end{equation}
with diagonal block-matrices $s_{ij} \in \mathbb{R}^{2 \times 2}$. It follows
\begin{align*}
S_{11} &= \tilde S_{11} V = 
\left(
\begin{array}{cc}
s_{11} U & s_{12} U \\ 0 & s_{22} U
\end{array}
\right) ,
\\
\Rightarrow \;\; S^{tr}_{11} S_{11} &= 
\left(
\begin{array}{cc}
U^{-1} s_{11} & 0 \\ U^{-1} s_{12} & U^{-1} s_{22}
\end{array}
\right) 
\left(
\begin{array}{cc}
s_{11} U & s_{12} U \\ 0 & s_{22} U
\end{array}
\right) 
=
\left(
\begin{array}{cc}
U^{-1} s_{11}^2 U & U^{-1} s_{11} s_{12} U \\
U^{-1} s_{12} s_{11} U & U^{-1} (s_{12}^2 + s_{22}^2) U 
\end{array}
\right) .
\end{align*}
Thm. \ref{thm:tam} tells us that we have to compute the LDL-Cholesky factorization of 
the positive definite 
\[
A_1 := U^{-1} s_{11}^2 U = 
\left(
\begin{array}{cc}
c & s \\ -s & c
\end{array}
\right)
\left(
\begin{array}{cc}
\beta_{1} c & - \beta_1 s \\
\beta_2 s & \beta_2 c
\end{array}
\right)
=
\left(
\begin{array}{cc}
\beta_1 c^2 + \beta_2 s^2 & cs (\beta_2 - \beta_1) \\
cs (\beta_2 - \beta_1) & \beta_1 s^2 + \beta_2 c^2
\end{array}
\right) .
\]
By lemma \ref{lem:ldl_2d}, using $\det(A_1) = \beta_1 \beta_2$, we obtain $A_1 = Q^{tr} H Q$ with
\begin{subequations}
\begin{align}
H^{1/2} &= \mathrm{diag}
\left(
\sqrt{\beta_1 c^2 + \beta_2 s^2}, \sqrt{\frac{\beta_1 \beta_2}{\beta_1 c^2 + \beta_2 s^2}}
\right) , \label{eq:h12_4d} \\
Q &=
\left(
\begin{array}{cc}
1 & cs(\beta_2 - \beta_1)/(\beta_1 c^2 + \beta_2 s^2) \\
0 & 1
\end{array}
\right) . \label{eq:q_4d}
\end{align}
In order to determine the remaining entry in the Iwasawa factor $N$, we compute (the details are left to the reader)
\begin{equation}
X := H^{-1} Q^{-tr} B_1 =
\left(
\begin{array}{cc}
- \frac{\alpha_1 c^2 + \alpha_2 s^2}{\beta_1 c^2 + \beta_2 s^2} & 
\frac{cs(\alpha_1 - \alpha_2)}{\beta_1 c^2 + \beta_2 s^2} \\
cs \left(\frac{\alpha_1}{\beta_1} - \frac{\alpha_2}{\beta_2}\right) &
-c^2 \frac{\alpha_2}{\beta_2} - s^2 \frac{\alpha_1}{\beta_1}
\end{array}
\right) . \label{eq:hqb_d4}
\end{equation}
\end{subequations}
By Eq. \eqref{eq:h12_4d} we see that the determinant of the $A$-factor of a given covariance matrix yields $\beta_1 \beta_2$. Moreover, we can recover the term $\sqrt{\beta_1 c^2 + \beta_2 s^2}$ and by Eq. \eqref{eq:q_4d} the $cs (\beta_2 - \beta_1)$. Eq. \eqref{eq:h12_4d} - \eqref{eq:hqb_d4} constitute a system of equations for the five parameters $\alpha_1$, $\alpha_2$, $\beta_1$, $\beta_2$ and $\psi$ of this model and more such equations may be obtained by computing the other decompositions of $S_{ij}$. In this way we can attempt to recover the entire set of parameters of this model from a given covariance matrix.\\\\
Since $N$ is symplectic, we have $N^{-1} = -J N^{tr} J$, i.e.
\begin{equation}
N^{-1} = \left(\begin{array}{cc} Q^{-1} & -X^{tr} \\ 0 & Q^{tr} \end{array}\right) .
\end{equation}
Using this equation, we can then compute the orthosymplectic $K$ via $K = S_{11} N^{-1} A^{-1}$. After some steps we find the block-diagonal form
\begin{equation}
K_{11} = \frac{1}{\sqrt{\beta_1 c^2 + \beta_2 s^2}} \left(\sqrt{\beta_1} c - \sqrt{\beta_2} s J_2^{\oplus 2} \right) .
\end{equation}
Hereby we attached the indices on $K$ to indicate that it emerges out of the map $S_{11}$. Note that $J_4 \not = J_2^{\oplus 2}$, so if $\psi \not = 0$, then this term is not in the form $D_1 + J_4 D_2$. This example therefore completes the statement in Conclusion \ref{con:iwa_opt}. What happens for the other cases $S_{12}$, $S_{21}$ and $S_{22}$? Since we expect a symmetric result for $S_{22}$ let us investigate the $S_{12}$ case. As we have already computed the Iwasawa factors $A$ and $N$, we do not need to recompute them again. $\tilde S_{12}$ has the form:
\begin{equation}
\tilde S_{12} = \left(
\begin{array}{cccc}
\sqrt{\beta_1} & 0 & -\alpha_1/\sqrt{\beta_1} & 0 \\
0 & 1/\sqrt{\gamma_2} & 0 & 0 \\
0 & 0 & 1/\sqrt{\beta_1} & 0 \\
0 & - \alpha_2/\sqrt{\gamma_2} & 0 & \sqrt{\gamma_2}
\end{array}
\right) .
\end{equation}
After further computations we arrive at the following orthosymplectic $K_{12}$:
\begin{equation}
K_{12} = \frac{1}{\sqrt{\beta_1 c^2 + \beta_2 s^2}} 
\left(\begin{array}{cccc}
c \sqrt{\beta_1} & - s \sqrt{\beta_2} & 0 & 0 \\
s / \sqrt{\gamma_2} & c \sqrt{\beta_1}/\sqrt{\beta_2 \gamma_2} & \alpha_2 s / \sqrt{\gamma_2} & \alpha_2 c \sqrt{\beta_1}/\sqrt{\beta_2 \gamma_2} \\
0 & 0 & c \sqrt{\beta_1} & -s \sqrt{\beta_2} \\
- \alpha_2 s / \sqrt{\gamma_2} & - \alpha_2 c \sqrt{\beta_1}/\sqrt{\beta_2 \gamma_2} & s / \sqrt{\gamma_2} & c \sqrt{\beta_1} / \sqrt{\beta_2 \gamma_2}
\end{array}\right) .
\end{equation}
We see that also $K_{12}$ turns into the standard form if we have no coupling. Moreover, there has to be a term $D_1 + J_4 D_2$ with $D_1^2 + D_2^2 = 1$ which transforms $K_{11}$ into $K_{12}$. In fact we find $K_{12} K_{11}^{tr} = D_1 + J_4 D_2$
with $D_1 = \mathrm{diag}(1, 1/\sqrt{\beta_2 \gamma_2}, 1, 1/\sqrt{\beta_2 \gamma_2})$ and
$D_2 = \mathrm{diag}(0, \alpha_2/\sqrt{\beta_2 \gamma_2}, 0, \alpha_2/\sqrt{\beta_2 \gamma_2})$. This corresponds to a rotation in the second plane by an angle of $\arctan(-\alpha_2)$. In general, any rotation in these planes can lead to valid $K$'s. Therefore
it requires a careful analysis of the covariance matrices involved in order to disentangle the sought coupling terms from that freedom.
\end{xmpl}

\section{Emittances from simulation and experiment} \label{subsec:emit_cov}

For large-scale simulations with many particles, the particle positions are usually not stored turn by turn, because
this will generate an enormous amount of data. What can be stored without generating too much data, however,
are the turn-by-turn covariance matrices of the distribution. In this section we will outline the connection to the well-established formula of Lapostolle. For completeness we will also include practical formulae to obtain emittances from experimental data.

\subsection{Simulation case}

By means of the map $V$, given by Thm. \ref{thm:s_diag_r} or \ref{con:iwa_opt} (from a covariance matrix) and utilized in Thm. \ref{thm:m_invariance}, we can
parameterize all symmetric invariants $G$ of $M$. As motivated in paragraph \ref{ssubs:class},
one may have the task to find, for a given covariance matrix $G$ an $M$-invariant covariance matrix
$G_*$ which is closest to $G$, so that we can apply Thm. \ref{thm:m_invariance}. Hereby we understand 
the distance between $G_*$ and $G$
as given by the Frobenius norm.
This can be formulated in a precise fashion as follows:\\\\
Let $\{e_j ; j \in \overline{2n}\}$ be the canonical basis of $\mathbb{R}^{2n}$. 
Introduce for $k \in \overline{n}$ the matrices
\begin{equation}
E_k := e_ke_{k}^{tr} + e_{n + k}e_{n + k}^{tr} \in \mathbb{R}^{2n \times 2n} ,
\end{equation}
i.e. $E_k$ consists of zeros besides its $(k, k)$ and $(n + k, n + k)$ components, which are one. 
Denote for brevity $W_k := V E_k V^{tr}$ with the notations of Thm. \ref{thm:m_invariance}. 
Then the problem stated above corresponds to the task of finding $\Lambda_k \in \mathbb{R}_{\geq 0}$
so that
\begin{equation}
f(\Lambda) := |G - \sum_{k = 1}^n \Lambda_k W_k |_F  \label{eq:matching}
\end{equation}
is minimized.
The $\Lambda_k$'s then correspond to the emittances, because the covariance 
matrix in Floquet-space has zero off-diagonal elements for independent variables and the determinant in the individual
Floquet-planes are therefore just $\Lambda_k^2$ (see also example \ref{xmpl:lapostolle} below). We remark that an expression as the sum in Eq. \eqref{eq:matching} also appears in Ref. \cite{bib:wolski2006}.
Since $h(G) := G - \sum_k \Lambda_k W_k$ is extremal at a given
point if and only if
$V^{-1} h(G) V^{-tr}$ is extremal at that point,
we obtain, by using the the symmetry of $\langle A, B \rangle_F := \mathrm{tr}(A^{tr} B)$ and $G_V := V^{-1} G V^{-tr}$:
\begin{align}
f(\Lambda)^2 &= |G_V|^2_F - 2 \sum_{k = 1}^n \Lambda_k \langle G_V, E_k \rangle_F + \sum_{k, l = 1}^n
\Lambda_k \Lambda_l \underbrace{\langle E_k, E_l\rangle_F}_{= \mathrm{tr}(E_k E_l) = 2 \delta_{kl}} \nonumber \\
&= |G_V|^2_F + 2 \sum_{k = 1}^n ( \Lambda_k^2 - \Lambda_k \langle G_V, E_k \rangle_F) \nonumber \\
&= |G_V|^2_F + 2 \sum_{k = 1}^n (\Lambda_k - \frac{1}{2} \langle G_V, E_k \rangle_F)^2 - \frac{1}{2} \sum_{k = 1}^n
\langle G_V, E_k \rangle_F^2 . \label{eq:gstar_g}
\end{align}
Hence, $f^2$ is minimal if and only if
\begin{equation}
\Lambda_k := \frac{1}{2} \langle V^{-1} G V^{-tr}, E_k \rangle_F . \label{eq:lambda_k}
\end{equation}
This process provides us with a closest symmetric invariant $G_* := \sum_k \Lambda_k W_k$.
Since $G$ was assumed to be positive (semi)definite, there exist $P$ so that $G = PP^{tr}$, i.e. 
$2 \Lambda_k = \mathrm{tr}(V^{-1} G V^{-tr} E_k) = \mathrm{tr}(V^{-1} P P^{tr} V^{-tr} E_k) 
= |Ae_k|^2 + |Ae_{n + k}|^2 \geq 0$ with $A := P^{tr} V^{-tr}$, so $G_*$ is again positive (semi)definite.
\begin{xmpl} \label{xmpl:lapostolle}
In the 2-dimensional case $n = 1$ we obtain the classical emittance definition by Lapostolle (cf. \cite{bib:lawsonetal1973})
as follows:
Assume that $G = \langle x x^{tr} \rangle$ is given and set $\Lambda_k$ according to Eq. \eqref{eq:lambda_k}.
With $z := V^{-1} x$ we get
\begin{equation}
2 \Lambda_1 = \langle G_V, E_1 \rangle_F = \mathrm{tr}(\langle V^{-1} x x^{tr} V^{-tr} \rangle E_1)
= \sum_{k = 1}^2 \mathrm{tr}(\langle e_k^{tr} z (e_k^{tr} z)^{tr} \rangle) = \langle z_1^2 \rangle + \langle z_2^2 \rangle .
\label{eq:2lambda1}
\end{equation}
On the other hand, by Eq. \eqref{eq:gstar_g} and $G_V = V^{-1} G V^{-tr} = \langle z z^{tr} \rangle$ we have
\begin{equation}
|G_* - G|_F^2 = |G_V|_F^2 - 2 \Lambda_1^2 , \label{eq:gstarg1}
\end{equation}
hereby 
\begin{equation}
|G_V|_F^2 = \mathrm{tr}(\langle z z^{tr} \rangle \langle z z^{tr} \rangle) = \sum_{ij} (\langle z_i z_j \rangle)^2
= \langle z_1^2 \rangle^2 + 2 (\langle z_1 z_2 \rangle)^2 + \langle z_2^2 \rangle^2 . \label{eq:gvnorm}
\end{equation}
Now combining Eqs. \eqref{eq:2lambda1}, \eqref{eq:gstarg1} and \eqref{eq:gvnorm} we obtain
\begin{equation}
2 \Lambda_1^2 = 4 \Lambda_1^2 - 2 \Lambda_1^2 = 2 \langle z_1^2 \rangle \langle z_2^2 \rangle - 2 (\langle z_1 z_2 \rangle)^2
+ |G - G_*|_F^2 ,
\end{equation}
and since $\mathrm{det}(V)^2 = 1$ we have $\mathrm{det}(G) = \mathrm{det}(\langle z z^{tr} \rangle)$, therefore we regain
the emittance of Lapostolle up to the Frobenius distance between $G$ and the $M$-invariant $G_*$:
\begin{equation}
\Lambda_1^2 = \langle x_1^2 \rangle \langle x_2^2 \rangle - (\langle x_1 x_2 \rangle)^2 + \frac{1}{2} |G - G_*|_F^2 .
\end{equation}
\end{xmpl}
\begin{xmpl}
In the special case of a single particle, i.e. if $G$ has the form
$G = x x^{tr}$ with $x \in \mathbb{R}^{2n}$, then $z^{tr} e_k := Ae_k = x^{tr} V^{-tr} e_k$
(compare above) and so we get with
$z := V^{-1}x$ the \textit{action} as a 'single-particle emittance':
\begin{equation}
2 \Lambda_k = z_k^2 + z_{n + k}^2 . 
\label{eq:sp_emit}
\end{equation}
We recall the standard 2D example of a linear transport map $M$ of Example \ref{xmpl:2d} (or found e.g. in Ref. \cite{bib:lee2012}):
\begin{equation}
M = \left(\begin{array}{cc}
\cos(\varphi) + \alpha \sin(\varphi) & \beta \sin(\varphi) \\
- \gamma \sin(\varphi) & \cos(\varphi) - \alpha \sin(\varphi)
\end{array}\right) ,
\end{equation}
where $\beta, \gamma \in \mathbb{R}_{> 0}$, $\varphi, \alpha \in \mathbb{R}$ and
$1 + \alpha^2 = \beta \gamma$.
Then one can show that with
\begin{equation}
V := \left(\begin{array}{cc}
1/\sqrt{\gamma} & -\alpha/\sqrt{\gamma} \\
0 & \sqrt{\gamma}
\end{array}\right)
\end{equation}
we have $V^{tr} J_2 V = J_2$ and
\begin{equation}
V^{-1} M V = \left(\begin{array}{cc}
\cos(\varphi) & \sin(\varphi) \\
- \sin(\varphi) & \cos(\varphi)
\end{array}\right) ,
\end{equation}
and with $x = (x_1, x_2)^{tr}$, $z = V^{-1} x$, 
we obtain for the action \eqref{eq:sp_emit}:
\begin{equation}
2 \Lambda_1 = z_1^2 + z_2^2 = \gamma x_1^2 + 2 \alpha x_1 x_2 + \beta x_2^2 .
\end{equation}
\end{xmpl} 

\subsection{Measurement case} \label{subs:c5_measc}

In this last paragraph we will summarize of how to obtain emittances from measured profile data in this general linear setting. We assume here that the coordinates are arranged in the form $x, y, z, p_x, p_y, p_z$. Denote for $i = 1, 2, 3$ 
$M_i \colon \mathbb{R}^6 \to \mathbb{R}^6$ the symplectic transport maps to the location of the scanners which measure our profiles (which are usually two wirescanners and a wall-current monitor)
and by $V \colon \mathbb{R}^6 \to \mathbb{R}^6$ the map from Floquet-space to ordinary phase space, 
which diagonalize the one-turn map $M$ by $R = V^{-1} M V$ according to Thm. \ref{thm:s_diag_r}.\\\\
Assume that $G$ is the covariance matrix belonging to a matched distribution. 
By Thm. \ref{thm:m_invariance} we have $G = VDV^{tr}$ with
$D = \mathrm{diag}(\Lambda, \Lambda)$ and $\Lambda = \mathrm{diag}(\lambda_1, \lambda_2, \lambda_3)$, i.e. $G$ can be interpreted as the image of a covariance matrix of a distribution in Floquet-space, 
in which the individual planes are uncorrelated, transported by the map $V$ to ordinary phase space.
For $i = 1, 2, 3$ consider the maps $T_i := e_i^{tr} M_i V$, where $e_i$ denotes the unit vector having a one in 
the $i$th position, so they project onto the spaces belonging to the $x$, $y$ and $z$ directions at the corresponding scanner locations. Now consider the linear map $E \colon \mathbb{R}^3 \to \mathbb{R}^3$ given by $(\lambda_1, \lambda_2, \lambda_3) \mapsto (T_1DT_1^{tr}, T_2DT_2^{tr}, T_3DT_3^{tr})$, so its matrix entries are
$E_{jk} = T_je_ke_k^{tr}T_j^{tr} + T_j e_{3 + k}e_{3 + k}^{tr}T_j^{tr}$. Since covariance matrices transport under linear maps in form of matrix congruence (see paragraph \ref{subs:prelims}), the image of this map can be identified with the second moments of the distribution $G$ at the corresponding scanner positions: $\langle x^2 \rangle$, $\langle y^2 \rangle$ and $\langle z^2 \rangle$. They are known from our experiments, hence $E^{-1}$ provides us with the emittances of the distribution.

\section{Conclusion}

We have examined in detail the close connection between linear normal form and covariance matrices belonging to a matched particle distribution. In fact, a linear normal form is contained in such a matrix: If the emittances are mutually distinguishable, then up to an $\mathrm{SO}(2)^n$-freedom (which can be understood as a rotation part related to the tunes) the entries of the normal form are uniquely determined. Furthermore, by means of the Iwasawa decomposition, we obtain a natural generalization of the optics $\beta$-functions and coupling terms, which complement, together with their relation \eqref{eq:beta_tensor} to the embedding coefficients $\beta_{ij}^{\hspace{.2cm} l}$, our picture in this linear scenario. 
In addition, we have provided the connection to the Courant-Snyder parameterization and the Lapostolle-emittance and summarized useful formulae for practical applications regarding simulations and experiments.

\section{Acknowledgments}

The author want to thank Prof. Dr. M. de Gosson and Dr. F. Schmidt for inspiring comments. This work has been sponsored by the Wolfgang Gentner Programme of the German Federal Ministry of Education and Research (grant no. 05E15CHA)

\section{Appendix}

\subsection{Preliminaries} \label{app:ec}

This part of the appendix is intended as a convenient reference of properties and notations which we used in the main text. Some of these facts are known in the literature but often scattered or not easy to find in a concise and self-contained fashion. We will begin with the proof of Thm. \ref{thm:cov_spsd}.\\\\
\textit{$G \in \mathbb{R}^{m \times m}$ is a covariance matrix if and only if $G$ is symmetric and positive semidefinite.}
\begin{proof}
'$\Rightarrow$' Symmetry is a consequence of $\langle x_i x_j \rangle = \langle x_j x_i \rangle$. Positive semidefiniteness
follows with $u^{tr} G u = u^{tr} \langle x x^{tr} \rangle u = \langle u^{tr} x x^{tr} u \rangle = \langle (u^{tr} x)^2 \rangle
\geq 0$.
'$\Leftarrow$' Since $G$ is symmetric, we can find, by Sylvester's law of inertia, an orthogonal matrix $Q$ 
and a diagonal matrix $D$ so that $D = Q^{tr}GQ$ hold. Since
$G$ is positive semidefinite, the diagonal entries $D_k$ are non-negative. Set $\sqrt{D}$ by taking the square root of these
diagonal elements, so that we obtain a Cholesky decomposition of $G$ by $G = Q\sqrt{D}\sqrt{D}^{tr}Q^{tr} = PP^{tr}$
with $P := Q \sqrt{D}$. Now take $m$ independent random variables $z_j$, i.e. $\langle z_i z_j \rangle = \delta_{ij}$ for $i, j = 1, ..., m$.
Set $x := Pz$. It follows $G = PP^{tr} = \langle Pz(Pz)^{tr} \rangle = \langle x x^{tr} \rangle$, so $G$ is a covariance matrix.
\end{proof}
\noindent Let us make a notation convention:
\begin{conv} \label{conv:ab}
The Symbol $\mathbb{K}$ means either 
$\mathbb{R}$ or $\mathbb{C}$.
$J$ denotes the symplectic structure
\[
J := J_n := \left(\begin{array}{cc} 0 & 1_n \\ -1_n &
    0 \end{array}\right) \in \mathbb{K}^{2n \times 2n} ,
\]
where $1_n$ denotes the identity matrix in $\mathbb{K}^{n \times n}$. The upper letter
'$H$' on a matrix means transposition and complex conjugation. For $x, y \in \mathbb{C}^{2n}$
we understand $\langle \cdot, \cdot \rangle$ as the sesquilinear form $\langle x, y \rangle := x^H y$
if nothing else is stated. We will sometimes use the abbreviation
$\overline{n} := \{1, ..., n\}$ for $n \in \mathbb{N}$. If $V$ is a vector space,
we denote its complexification by $V_\mathbb{C}$ and for $M \colon V \to V$, we sometimes
denote its action onto $V_\mathbb{C}$ by $M_\mathbb{C}$. However, this notion will be dropped whenever the context is clear.\\\\
Let $M \in \mathrm{Sp}(2n; \mathbb{R})$ be diagonalizable with mutually distinguishable eigenvalues. We denote
by $\{a_j ; j \in \overline{2n}\}$
a fixed basis of eigenvectors, where $a_j \in \mathbb{C}^{2n}$ belongs to the eigenvalue $\lambda_j \in \mathbb{C}$. 
Because $M$ is real, $\bar \lambda_j$ is the eigenvalue of the eigenvector $\bar a_j$ of $M$.
We have for all $i, j \in \overline{2n}$:
\begin{equation}
\bar \lambda_i \langle a_i, J a_j \rangle =
\langle M a_i, J a_j \rangle = \langle a_i, M^{tr} J a_j \rangle = \langle a_i, J M^{-1} a_j \rangle 
= 1/\lambda_j \langle a_i, J a_j \rangle = \bar \lambda_j \langle a_i, J a_j \rangle, \label{eq:ai_j_aj}
\end{equation}
so we conclude, since all eigenvalues are mutually distinguishable, that if $i \not = j$, then $a_i$ and $J a_j$ are orthogonal.
Because of this orthogonality, the fact that $\{Ja_i\}$ is a basis of $\mathbb{C}^{2n}$ and $\langle \cdot, \cdot \rangle$ is non-degenerate, it must
hold that $\forall i \colon \langle a_i, J a_i \rangle \not = 0$,
and these values are purely imaginary, which follows by $\langle x, y \rangle = \overline{\langle y, x \rangle}$ with $J^{tr} = -J$.
Set $i \sigma_j := \langle a_j, J a_j \rangle$ with $\sigma_j \in \mathbb{R} \backslash \{0\}$.
Since $- i \sigma_j = \overline{\langle a_j, J a_j \rangle} = \langle \bar a_j, J \bar a_j \rangle$, we can
choose a representation system $\{j_1, j_2, ..., j_n\} \subset \overline{2n}$ of the equivalence
relation introduced in Dfn. \ref{dfn:equirel}, so that $\forall k \in \overline{n} \colon \sigma_{j_k} = 1$ hold.\\\\
The eigenvalues of $M^{tr}$ coincide with $M$, and if we set $b_i := Ja_i$, then
\[
M^{tr}b_i = - JM^{-1}Jb_i = JM^{-1}a_i = \bar \lambda_i Ja_i = \bar
\lambda_i b_i ,
\]
i.e. $b_i$ is an eigenvector or $M^{tr}$ with respect to $\bar \lambda_i$. Let us summarize:
\begin{enumerate}
\item If we speak of an eigenvector $b_i$ of $M^{tr}$ we will always understand $b_i := J a_i$ for a given (and fixed) system of eigenvectors
$\{a_i, i \in \overline{2n}\}$ of $M$.
\item From the properties of $M$ we always have $i \not = j \Rightarrow \langle a_i, J a_j \rangle = 0$.
\item There is a subset $\{j_1, ..., j_n\} \subset \overline{2n}$ so that 
$\forall k \in \overline{n} \colon \langle a_{j_k}, J a_{j_k} \rangle = i \sigma_k$ with $\mathbb{R} \ni \sigma_k > 0$ holds.
\end{enumerate}
\end{conv}
\noindent Proofs of the next statements are straightforward. They are required in the proof of the linear normal form Thm. \ref{thm:s_diag_r}.
\begin{prop} \label{prop:evon}
Let $V$ be a $\mathbb{K}$-vector space and $M \colon V \to V$ linear.
Then all eigenvectors of $M$ belonging to mutually distinguishable eigenvalues are linearly
independent.
\end{prop}
\begin{dfn} \label{dfn:equirel}
Let $V$ be an $\mathbb{R}$-vector space and $M \colon V \to V$ linear.
Assume that $M_\mathbb{C}$ has mutually distinguishable eigenvalues $\{\lambda_j \in \mathbb{C}; j \in I_0 \}$.
Since $M$ is real, it also admits the complex conjugate eigenvalues. So we can introduce on $I_0$
the equivalence relation $i \sim j :\Leftrightarrow \lambda_i = \bar \lambda_j$.
We denote the equivalence class of $j \in I_0$ by $[j]$. They constitute of pairs of indices.
\end{dfn}
\begin{prop} \label{prop:re_2basis}
Let $V$ be a real vector space and $M \colon V \to V$ linear. Assume that $M_\mathbb{C}$ has
mutually distinguishable eigenvalues
$\{\lambda_j \in \mathbb{C} ; j \in I_0 \}$ with corresponding 
eigenvectors $a_j = x_j + i y_j \in V_\mathbb{C}$
so that $x_j, y_j \in V$.
Then $\{x_{j_1}, y_{j_1}, x_{j_2}, y_{j_2}, ...\}$ are linearly independent in $V_\mathbb{C}$ for every representation system
$\{j_1, j_2, ...\}$ of the equivalence relation introduced in Dfn. \ref{dfn:equirel}.
\end{prop}

\begin{prop} \label{prop:span_inv}
Let $V$ be a real vector space and $M \colon V \to V$ linear. Let $\lambda \in \mathbb{C}$, be
an eigenvalue of $M_\mathbb{C}$ with eigenvector
$a = x + iy \in V_\mathbb{C}$ so that $x, y \in V$.
Then the $\mathbb{C}$-vector space spanned by $x$ and $y$ in $V_\mathbb{C}$ is $M$-invariant and it holds
\begin{equation}
\forall \alpha, \beta \in \mathbb{C} \colon \;\; M(\alpha x + \beta y) =
(\alpha \lambda_R + \beta \lambda_I)x + (\beta \lambda_R - \alpha \lambda_I)y ,
\end{equation}
where $\lambda_R := \mathrm{Re}(\lambda)$ and $\lambda_I := \mathrm{Im}(\lambda)$ are the real and imaginary parts
of $\lambda$.
\end{prop}
The next map emerged rather often in our programs as well as in some formulae, so that we found it useful to write it down as reference. It appears whenever we had to switch between the $(x, p_x, y, p_y, ..)$ phase-space notation to a block notation of the form $(x, y, ..., p_x, p_y, ...)$. But we also used this (non-symplectic) isomorphism in Thm. \ref{thm:m_invariance} to transport a statement regarding $2 \times 2$ matrices to a statement regarding block matrices and where it is useful to keep track that changing the notation does not have any effect on the symplecticity of the result.
\begin{dfnsz} \label{dfnsz:oplus}
Let
  $A \in \mathbb{K}^{2n \times 2n}$ and $B \in \mathbb{K}^{2m \times 2m}$ be two block-matrices of the form
\begin{align*}
A &= 
\left(
\begin{array}{cc}
A_{11} & A_{12} \\
A_{21} & A_{22}
\end{array}
\right) ,
& B &= 
\left(
\begin{array}{cc}
B_{11} & B_{12} \\
B_{21} & B_{22}
\end{array}
\right) ,
\end{align*} 
with $A_{ij} \in \mathbb{K}^{n \times n}$ and $B_{ij} \in \mathbb{K}^{m \times m}$ respectively.
Define $A \diamond B \in \mathbb{K}^{2(n + m) \times 2(n + m)}$ as
  \[
  A \diamond B := \left(
    \begin{array}{cccc}
      A_{11} & 0 & A_{12} & 0 \\
      0 & B_{11} & 0 & B_{12} \\
      A_{21} & 0 & A_{22} & 0 \\
      0 & B_{21} & 0 & B_{22} 
    \end{array}
  \right) .
  \]
It holds
\begin{enumerate}
\item $\diamond$ is bilinear.
\item  If $A, C \in \mathbb{K}^{2n \times 2n}$ and $B, D \in \mathbb{K}^{2m \times 2m}$, then
$(A \diamond B)(C \diamond D) = (AC \diamond BD)$.
\item For $A \in \mathbb{K}^{2n \times 2n}$, $B \in \mathbb{K}^{2m \times 2m}$ and
$C \in \mathbb{K}^{2k \times 2k}$ associativity holds: $A \diamond (B \diamond C) = (A \diamond B) \diamond C$.
\item  $(A \diamond B)^{tr} = A^{tr} \diamond B^{tr}$.
\item  $\det(A \diamond B) = \det(A) \det(B)$.
\item $J_n \diamond J_m = J_{n + m}$.
\end{enumerate}
\end{dfnsz}

\begin{rmk} \label{rmk:dia_oplus}
For convenience we may want to recast $A_1 \diamond \cdots \diamond A_n$ for
$A_k \in \mathbb{K}^{2 \times 2}$, $k \in \overline{n}$, in block-diagonal form, which we will denote
by the symbol $\oplus$ as $A_1 \oplus \cdots \oplus A_n
= \mathrm{diag}(A_1, ..., A_n)$.
For this purpose we introduce the following orthogonal operator $T \in \mathbb{K}^{2n \times 2n}$ (and
also provide its inverse) on the canonical basis $\{e_j ; j \in \overline{2n}\}$ of $\mathbb{K}^{2n}$:
\begin{align*}
T(e_j) &:= \left\{
\begin{array}{ll}
e_{(j + 1)/2} & \text{if $j$ is odd,} \\
e_{n + j/2} & \text{else,}
\end{array}\right.  &
T^{-1}(e_j) &= \left\{
\begin{array}{ll}
e_{2j - 1} & \text{if } j \in \overline{n}, \\
e_{2(j - n)} & \text{else.}
\end{array}\right.
\end{align*}
Then it holds $T^{tr} (A_1 \diamond \cdots \diamond A_n) T = A_1 \oplus \cdots \oplus A_n$.
\end{rmk}

\noindent A proof of Thm. \ref{thm:s_diag_r} can be found for example in Ref. \cite{bib:dragt2018}, but because of its importance and also because of certain details in the construction of the map we will recall it here:
\begin{proof}
Let us rescale the $a_j$'s by $\sqrt{2/|\sigma_j|}$, where $\sigma_j$ is given according to Conv. \ref{conv:ab}.
So without loss of generality we consider eigenvectors so that $\sigma_j = \pm 2$ hold.
By Conv. \ref{conv:ab} there is a representation system
$\{j_1, j_2, ..., j_n\} \subset \overline{2n}$ so that $\forall k \in \overline{n} \colon \sigma_{j_k} = 2$ holds.
According to Prop. \ref{prop:re_2basis}, we obtain a corresponding
real basis $\{x_1, x_2, ..., x_n, y_1, y_2, ..., y_n \}$ of $\mathbb{C}^{2n}$
with $a_{j_k} = x_k + i y_k$. We thus have by construction $\forall k, l \in \overline{n}$:
\begin{subequations}
\begin{align}
i 2 \delta_{kl} &= \langle a_{j_k}, J a_{j_l} \rangle = \langle x_k + i y_k, J x_l + i J y_l \rangle
= \langle x_k, J x_l \rangle + \langle y_k, J y_l \rangle - i \langle y_k, J x_l \rangle + i \langle x_k, J y_l \rangle , \label{eq:i2dkl_a} \\
0 &= \langle a_{j_k}, J \bar a_{j_l} \rangle = \langle x_k + i y_k, J x_l - i J y_l \rangle 
= \langle x_k, J x_l \rangle - \langle y_k, J y_l \rangle - i \langle y_k, J x_l \rangle - i \langle x_k, J y_l \rangle . \label{eq:i2dkl_b}
\end{align}
\end{subequations}
Therefore $\forall k, l \in \overline{n}$:
\begin{subequations}
\begin{align}
\langle x_k, J x_l \rangle &= 0 , \\
\langle y_k, J y_l \rangle &= 0 , \\
\langle x_k, J y_l \rangle &= \delta_{kl} ,
\end{align}
\end{subequations}
and so the linear map $V \colon \mathbb{R}^{2n} \to \mathbb{R}^{2n}$ defined on the canonical basis 
$\{e_j, j \in \overline{2n}\}$ of $\mathbb{R}^{2n}$ via
\begin{equation}
V(e_j) := \left\{\begin{array}{ll}
x_j & \text{if } j \in \overline{n} , \\
y_{j - n} & \text{else,} 
\end{array}\right.
\end{equation}
is symplectic. By Prop. \ref{prop:span_inv} we know that for $k \in \overline{n}$ the planes
$\tilde E_k := \mathrm{span}_\mathbb{R}\{V(e_k), V(e_{n + k})\} \subset \mathbb{R}^{2n}$ are
$M$-invariant; i.e.
if $(\alpha, \beta)^{tr} \in \mathbb{R}^2$ denote the vector $\alpha V(e_k) + \beta V(e_{n + k}) \in \tilde E_k$, then
with $\lambda := \lambda_{j_k} \in \mathbb{R}$:
\begin{equation}
M |_{\tilde E_k} 
\left(
\begin{array}{c}
\alpha \\ \beta
\end{array}
\right)
= 
\left(
\begin{array}{c}
\alpha \lambda_R + \beta \lambda_I \\
\beta \lambda_R - \alpha \lambda_I
\end{array}
\right)
=
\left(
\begin{array}{cc}
\lambda_R & \lambda_I \\
- \lambda_I & \lambda_R
\end{array}
\right)
\left(
\begin{array}{c}
\alpha \\ \beta
\end{array}
\right) . \label{eq:m_inv_block}
\end{equation}
By assumption $M$ has eigenvalues on the unit circle, $\lambda_R^2 + \lambda_I^2 = 1$, 
so $V^{-1} M V$ has the properties as claimed.
\end{proof}
\begin{lem} \label{lem:r_invariance}
Let $R = R_1 \oplus \cdots \oplus R_n \in \mathbb{K}^{2n \times 2n}$ be block diagonal with
\[
R_i = \left(\begin{array}{cc}
\cos(\varphi_i) & - \sin(\varphi_i) \\
\sin(\varphi_i) & \cos(\varphi_i)
\end{array}\right) ,
\]
with mutually disjoint $\varphi_i$, where $\varphi_i \in ]0, \pi[$. Let $B \in \mathbb{K}^{2n \times 2n}$ 
be given. Then
\begin{equation}
RB = BR \;\; \Leftrightarrow \;\; B = D_1 + J_2^{\oplus n} D_2 ,
\end{equation}
where $D_j$ are diagonal matrices
and of the form $\mathrm{diag}(b_1, b_1, b_2, b_2, 
..., b_n, b_n)$. 
\end{lem}
\begin{proof}
Since the $R_i$'s are orthogonal and commute with $J_2$, the '$\Leftarrow$' direction is clear, so let us prove the '$\Rightarrow$'
direction. Decompose $B$ into a symmetric and an antisymmetric part $B = \frac{1}{2}(B + B^{tr}) + \frac{1}{2}(B - B^{tr}) 
=: S + A$. Since $R$ is antisymmetric, it follows
\begin{subequations}
\begin{align}
RS + RA &= R(S + A) = (S + A)R = SR + AR , \label{eq:rsra} \\
\mathrm{tr} \Rightarrow -SR + AR &= -RS + RA . \label{eq:srra}
\end{align}
\end{subequations}
By adding and subtracting Eqs. \eqref{eq:rsra} and \eqref{eq:srra} we obtain the equivalent conditions
$RS = SR$ and $AR = RA$. So let us assume for a moment that $B$ is (anti)symmetric.\\\\
Condition $BR = RB$ in block indices reads: $\forall i, j \in \overline{n} \colon B_{ij} R_j = R_i B_{ij}$. 
In particular, by exchanging
indices, we can write down the following two equations for every $i$ and $j$:
\begin{subequations}
\begin{align}
B_{ij} R_j &= R_i B_{ij} , \label{eq:brrb1} \\
B_{ji} R_i &= R_j B_{ji} .
\end{align}
\end{subequations}
From these two equations we obtain
\begin{equation}
R_j B_{ji} B_{ij} = B_{ji} R_i B_{ij} = B_{ji} B_{ij} R_j . \label{eq:rbb}
\end{equation}
Since $B$ is (anti)symmetric, $B_{ij} = \pm B_{ji}^{tr}$, and therefore with the positive semidefinite symmetric
$L_{ij} := B_{ij}^{tr} B_{ij}$, Eq. \eqref{eq:rbb} reads
\begin{equation}
R_j L_{ij} = L_{ij} R_j . \label{eq:rjlij}
\end{equation}
Now write $L_{ij}$ in terms of a $2 \times 2$ symmetric matrix
\[
L_{ij} =: \left(\begin{array}{cc}
r & d \\
d & f
\end{array}
\right) . 
\]
For brevity denote $c_j := \cos(\varphi_j)$ and $s_j := \sin(\varphi_j)$. Then Eq. \eqref{eq:rjlij}
reads
\begin{align*}
\left(
\begin{array}{cc}
r c_j - d s_j & d c_j - f s_j \\
r s_j + d c_j & d s_j + f c_j
\end{array}
\right) =
\left(\begin{array}{cc}
c_j & -s_j \\
s_j & c_j
\end{array}\right)
\left(\begin{array}{cc}
r & d \\
d & f
\end{array}\right) &=
\left(\begin{array}{cc}
r & d \\
d & f
\end{array}\right)
\left(\begin{array}{cc}
c_j & -s_j \\
s_j & c_j
\end{array}\right)
= \left(\begin{array}{cc}
r c_j + d s_j & - r s_j + d c_j \\
d c_j + f s_j & - d s_j + f c_j 
\end{array}\right) .
\end{align*}
Since $s_j \not = 0$ it follows from the $(1, 1)$ or $(2, 2)$-component that $d = 0$. Inserting
this into e.g. the $(1, 2)$-component we see that $r = f$ must hold and from $L_{ij} \geq 0$, it follows
that $r \geq 0$. We now attach the indices $i$ and $j$ on $r$. We conclude that 
$r_{ij} \boldsymbol{1}_2 = B_{ij}^{tr} B_{ij}$, so if $B_{ij} \not = 0$,
$C_{ij} := B_{ij}/\sqrt{r_{ij}}$ must be orthogonal.\\\\
In the case that $B_{ij} \not = 0$ there are now two options: Either $\det(C_{ij}) = 1$ or $\det(C_{ij}) = -1$. In the first case, $C_{ij}$ and therefore $B_{ij}$ commutes with $R_j$ and we get together with Eq. \eqref{eq:brrb1}
\begin{equation}
R_j B_{ij} = B_{ij} R_j = R_i B_{ij} .
\end{equation}
Now the second case. By multiplication of $BR = RB$ with the inverse of $R$, also $BR(-\varphi) = R(-\varphi)B$ hold. In this version, Eq. \eqref{eq:brrb1} reads
\begin{equation}
B_{ij} R_j(-\varphi_j) = R_i(-\varphi_i) B_{ij} .
\end{equation}
Let $W := \left(\begin{array}{cc} 0 & 1 \\ 1 & 0 \end{array}\right)$. Then $C_{ij} W$ and therefore $R_{ij} W$ commutes with $R_j$, so we get
\begin{equation}
R_j B_{ij} = R_j B_{ij} W W = B_{ij} W R_j W = B_{ij} R_j(-\varphi_j) = R_i(-\varphi_i) B_{ij} .
\end{equation}
Hence in both cases a equation of the form $R_j = R_i(\pm \varphi_i)$ follows, so $\varphi_i =
\pm \varphi_j$ and therefore, by construction, this is only possible if $\varphi_i = +\varphi_j$ and so $i = j$. Conversely we have shown that if $i \not = j$, then $B_{ij} = 0$, so $B$ must be block-diagonal and its diagonal blocks must have positive determinant.\\\\
Recall that we assumed that $B$ was (anti)symmetric. In the symmetric case, since
$B_{ii}^{tr} = B_{ii}$, and at the same time $B_{ii}/\sqrt{r_{ii}}$ is orthogonal (if $B_{ii} \not = 0$), the individual off-diagonal
elements must vanish and therefore $B$ has a diagonal form as claimed. In the antisymmetric case 
$B_{ii}^{tr} = - B_{ii}$, so its diagonal entries vanish and it has the form 
$J_2 \mathrm{diag}(b, b)$.
\end{proof}

\subsection{A basis for symmetric invariant matrices} \label{subs:symbasis}

Symmetric matrices which are $M$-congruent invariant, where $M$ is diagonalizable with mutually distinguishable eigenvalues, can be given a basis built out of the eigenvalues of $M$ itself. This very useful result, as stated in Cor. \ref{cor:real_sym_inv_basis}, is used in Ref. \cite{bib:nash2006} in order to find matched distributions near coupled synchrobetatron resonances. As we shall see at the end of the next paragraph, this decomposition is linked to the diagonalization of invariant covariance matrices by linear normal form, which is used in some of the other references mentioned in the introduction. For convenience, we change our notation to $M^{tr}$ \textit{only in this paragraph \ref{subs:symbasis}}, as otherwise we would have to attach many minus signs on the maps. 
Assume that $G$ is a symmetric invariant and define $X := JG$. Since $M^{tr}$ is symplectic, $M J M^{tr} = J$, condition \ref{dfn:m_inv} can be recast as
\[
X = JG = J M^{tr} G M = M^{-1} J G M = M^{-1} X M,
\]
and so 
\begin{equation}
X = M^{-1} X M .  \label{eq:adx_x}
\end{equation}
Denote by $\mathfrak{g}$ the Lie-algebra of $\mathrm{Sp}(2n; \mathbb{R})$. One can show that $\mathfrak{g}$ can be characterized as
$\mathfrak{g} = \{ X \in \mathbb{R}^{2n \times 2n} ;
J X + X^{tr} J = 0\}$. The elements of this semisimple Lie-algebra are called \textit{Hamiltonian matrices}. Now observe that since
$G$ is symmetric,
\begin{equation}
J (JG) + (JG)^{tr} J = - G - G^{tr} J^2 = 0 ,
\end{equation}
so together with Eq. \eqref{eq:adx_x} we have the characterization $X \in \mathfrak{g}$ with $\mathrm{Ad}(M)X = X$. If
$X \in \mathfrak{g}$, then conversely $JX$ is symmetric: $(JX)^{tr} = -X^{tr}J = JX$. Let $\mathfrak{h} := \{X \in \mathfrak{g} ; \mathrm{Ad}(M)X = X\}
= \{ X \in \mathfrak{g} ; [M, X] = 0 \}$.
$\mathfrak{h}$ is a Lie-subalgebra of $\mathfrak{g}$ since $\mathrm{Ad}(M)$ enters $[\cdot, \cdot]$ in both entries.
Condition \ref{dfn:m_inv} therefore essentially
means that we are considering elements $X = JG$ of the Lie-subalgebra $\mathfrak{h}$ of $\mathfrak{g}$ and so that
$G$ may in addition be positive semidefinite.
\begin{lem} \label{lem:eigad}
Let $A$ be a diagonalizable real or complex $N \times N$-matrix with eigenvalues $\nu_1, ..., \nu_N \in \mathbb{C}$. 
Then $\mathrm{ad}(A) = [A, \cdot]$ has the $N^2$ eigenvalues
$\tau_{ij} := \nu_i - \nu_j \in \mathbb{C}$ for $1 \leq i, j \leq N$ and the corresponding eigenmatrices
$V_{ij} \in \mathbb{C}^{n \times n}$
to $\tau_{ij}$ have the form
\[
V_{ij} = u_i v_j^H ,
\]
where $u_i$ and $v_j$ are eigenvectors of $A$ and $A^H$ with respect to $\nu_i$ and $\bar \nu_j$.
\end{lem}
\begin{proof}
Since $A^H$ has the eigenvalues $\bar \nu_1, ..., \bar \nu_N$, let $A u_i = \nu_i u_i$ and $A^H v_i = \bar \nu_i v_i$. It follows
\[
[A, V_{ij}] = A u_i v_j^H - u_i v_j^H A = \nu_i u_i v_j^H - u_i (A^H v_j)^H = \nu_i u_i v_j^H - \nu_j u_i v_j^H  = \tau_{ij} V_{ij} .
\]
\end{proof}

\begin{prop} \label{prop:h_abelian}
Let $M \in \mathrm{Sp}(2n; \mathbb{R})$ be diagonalizable with mutually distinguishable eigenvalues and
denote by $\mathfrak{g}$ the Lie-algebra of $\mathrm{Sp}(2n; \mathbb{R})$ and
$\mathfrak{h} := \{X \in \mathfrak{g} ; [M, X] = 0\}$. Then $\mathfrak{h}$ is abelian.
\end{prop}
\begin{proof}
By application of lemma \ref{lem:eigad} to the real matrix $M$ we obtain the $(2n)^2$ eigenmatrices $V_{ij} = a_i b_j^H$ of $\mathrm{ad}(M)$. By Conv. \ref{conv:ab}, the entry $b_k^H V_{ij} a_l$ is zero if $k \not = i$ or $j \not = l$ and otherwise not, and so these eigenmatrices are linearly independent. Therefore $\mathfrak{m} := \{X \in \mathbb{C}^{2n \times 2n} ;
[M, X] = 0\}$ is spanned by
the matrices $B_i := V_{ii} = a_i b_i^H$, which
are the eigenmatrices belonging to the eigenvalue $0$.
By Conv. \ref{conv:ab} we have, since the eigenvalues are mutually distinguishable, for $i \not = j$:
$\langle a_i, b_j \rangle = 0$.
It follows $B_i B_j = a_i b_i^H a_j b_j^H = 0$ if $i \not = j$, showing that $\mathfrak{m}$ is abelian, especially
its sub Lie-algebras $\mathfrak{h} \subset \mathfrak{h}_\mathbb{C} \subset \mathfrak{m}$.
\end{proof}
\begin{prop} \label{prop:hc_basis}
There exist a basis $\{C_k; k \in \pm \bar n\}$ of linearly independent and $J$-unitary vectors of $\mathfrak{m}$ with respect to the
hermitian bilinear form $\langle X, Y \rangle_J := - \mathrm{tr}(J X^H J Y)$, i.e.
\begin{equation}
\forall k, l \in \pm \bar n \colon\langle C_k, C_l \rangle_J = \delta_{kl}.
\end{equation}
\end{prop}
\begin{proof} 
Let $\{j_1, ..., j_n\} \subset \overline{2n}$ be a representation system according to Conv. \ref{conv:ab} and rescale the $a_{j_k}$'s so that
$\langle a_{j_k}, J a_{j_k} \rangle = i$ holds. Set $C_k := a_{j_k} b_{j_k}^H$ and use the notation $-j_k$
for the other element in the equivalence class $[j_k]$. Correspondingly let $C_{-k} := a_{-j_k} b_{-j_k}^H = \overline{C_k}$ . By the proof of Prop.
\ref{prop:h_abelian} $\{C_k ; k \in \pm \bar n\}$ constitute a basis of $\mathfrak{m}$. We have
\begin{align}
\langle C_k, C_{\pm l} \rangle_J &= - \mathrm{tr}(J b_{j_k} a_{j_k}^H J a_{\pm j_l} b_{\pm j_l}^H) 
= - \mathrm{tr}(b_{\pm j_l}^H J b_{j_k} a_{j_k}^H J a_{\pm j_l}) =
- 2n \langle b_{\pm j_l}, J b_{j_k} \rangle \langle a_{j_k}, J a_{\pm j_l} \rangle \nonumber \\
&= -2n \langle Ja_{\pm j_l}, J^2 a_{j_k} \rangle \langle a_{j_k}, J a_{\pm j_l} \rangle = 2n
| \langle a_{j_k}, J a_{\pm j_l} \rangle |^2 .
\end{align}
We see that by an additional rescaling by $\sqrt{2n}$ of the $C_k$'s we obtain $J$-unitarity.
\end{proof}
\begin{cor}[See Ref. \cite{bib:nash2006}]
The set of all real symmetric invariants of $M$, where $M$ is diagonalizable with mutually distinguishable
eigenvalues, is spanned by the $n$ linearly independent matrices
$J (a_{j_k} a_{j_k}^H + \bar a_{j_k} a_{j_k}^{tr}) J$, where $\{j_1, j_2, ..., j_n\} \subset \overline{2n}$
is a representation system according to Conv. \ref{conv:ab}.
\end{cor}
\begin{proof}
From Prop. \ref{prop:hc_basis} we saw that $\mathfrak{m}$ admits a $J$-unitary basis $\{C_k; k \in
\pm \bar n\}$. 
Let $X = \sum_k X_k C_k \in \mathfrak{m}$ be given with
real- and imaginary parts $X_k = X^R_k + i X^I_k$, $C_k = C^R_k + i C^I_k$. Then
\begin{align}
& \bar X = \sum_k \bar X_k \bar C_k = \sum_{k} \bar X_k C_{-k} , \nonumber \\
& \text{$\bar X = X$ and $J$-unitarity of $C_k$} \;\; \Rightarrow \;\; X_k = \bar X_{-k} \;\;
\Leftrightarrow \;\;
\left\{\begin{array}{rcl} 
X^R_k &=& X^R_{-k} , \\
X^I_k &=& - X^I_{-k} . 
\end{array}\right. \label{eq:xicr_xrci2}
\end{align}
It follows if $X$ is real
\begin{equation}
X = \sum_{k = 1}^n X^R_k (C_k^R + C_{-k}^R) - \sum_{k = 1}^n X_k^I (C_k^I - C_{-k}^I) =
\sum_{k = 1}^n 2 X^R_k C_k^R - \sum_{k = 1}^n 2 X_k^I C_k^I , \label{eq:x_xrcr_xici2}
\end{equation}
and so $X$ is real if and only if it can be represented as a sum of the real matrices 
$C_k^R$ and $C_k^I$
Now observe that $2 J C_k^R = -J (a_{j_k} a^H_{j_k} + \bar a_{j_k} a_{j_k}^{tr}) J$
is symmetric, while $2 i J C_k^I = - J(a_{j_k} a^H_{j_k} - \bar a_{j_k} a^{tr}_{j_k}) J$ is antisymmetric. So $X \in \mathfrak{h}$
if and only if $X$ is represented as a sum of the $J$-orthogonal $C_k^R$ for $k \in \overline{n}$.
\end{proof}
\noindent
As we shall see, the result \ref{cor:real_sym_inv_basis} can equivalent be obtained by means of linear normal form, which will be discussed now.

\subsection{Proof of Thm. \ref{thm:mt} Part 1} \label{subsec:evev} 

For the proof Thm. \ref{thm:mt} it is more convenient to change the notation in its claim to $M^{tr}$ (as otherwise we get many minus signs in the exponents). 
A similar assertion can be found in Ref. \cite{bib:degosson2006}, however the proof unfortunately 
contained a mistake \cite{bib:degosson_pc}. We did not found an alternative proof. 
For the next part we will drop the indices 1 and 2 for convenience.\\\\
Since $G$ is positive definite, there exist a Cholesky-decomposition of $G$ in the
form $G = P^H P$, with
invertible $P \in \mathbb{C}^{2n \times 2n}$ (in fact, $P$ is real but for convenience we keep the complex notation). Then the invariance condition \ref{dfn:m_inv}
can be rewritten as
\[
(PMP^{-1})^H PMP^{-1} = 1 ,
\]
i.e. $U := PMP^{-1}$ is unitary. Since
$\mathrm{det}(M - \lambda) = \mathrm{det}(PMP^{-1} - \lambda)$, $U$
must have the same eigenvalues as $M$. Let $\{v_i \in \mathbb{C}^{2n}; i \in \overline{2n}\}$
be a basis of eigenvectors of $U$ with respect to the eigenvalue
$\lambda_i \in \mathbb{C}$, i.e. $U v_i = \lambda_i v_i$.\\\\
We have $U^H v_i = U^{-1} v_i = \lambda_i^{-1} v_i$ and therefore for
every $i$ and $j$:
\begin{equation}
  \bar \lambda_i \langle v_i, v_j \rangle = \langle \lambda_i v_i, v_j \rangle = \langle U v_i, v_j \rangle = \langle v_i, U^H v_j \rangle = \lambda_j^{-1} \langle v_i, v_j \rangle .
  \label{eq:barlilj}
\end{equation}
in particular ($i = j$) it follows that $|\lambda_i|^2 = 1$ for every
$i$, i.e. all eigenvalues lay on the unit circle. Furthermore if
$i \not = j$, then by assumption
$\bar \lambda_i \not = \bar \lambda_j$. Consequently we
must have $\langle v_i, v_j \rangle = 0$ in this case.\\\\
Let $a_i$ be the eigenvectors of $M$ with respect to $\lambda_i$ and
$b_i = Ja_i$ the eigenvectors of $M^{tr}$ with respect to $\bar \lambda_i$.
Since the eigenvalues are mutually distinguishable by our assumption,
they are connected to the orthogonal eigenvectors $v_i$ of $U$ (see
above) as follows:
\begin{subequations}
  \begin{align}
    \lambda_i v_i = U v_i = P M P^{-1} v_i &\Rightarrow P^{-1} v_i = \alpha_i a_i , \label{eq:vipai} \\
    \bar \lambda_i v_i = \lambda_i^{-1} v_i = U^H v_i = (P M P^{-1})^H v_i = P^{-H} M^{tr} P^H v_i &\Rightarrow P^H v_i
                                                                                                     = \beta_i b_i , \label{eq:bvipbi}
  \end{align}
\end{subequations}
with $\alpha_i, \beta_i \in \mathbb{C} \backslash \{0\}$ dependent on
$v_i$ (and therefore unspecified yet). It follows
\begin{subequations}
  \begin{align}
    & \langle v_i, v_j \rangle = \bar \alpha_i \alpha_j (Pa_i)^H Pa_j = \bar \alpha_i \alpha_j a_i^H P^H P a_j =
      \bar \alpha_i \alpha_j a_i^H G a_j , \label{eq:vv1} \\
    & \bar \alpha_i \beta_j \langle a_i, b_j \rangle = (P^{-1}v_i)^H P^H v_j = v_i^H P^{-H} P^H v_j = \langle v_i, v_j \rangle . \label{eq:ab}
  \end{align}
\end{subequations}
If we fix an index $i$, then for every $j$:
$\beta_j \langle a_i, b_j \rangle = \alpha_j \langle a_i, G a_j
\rangle$, thus
\[
Ga_j = \beta_j / \alpha_j b_j ,
\]
and so $JGa_j = -\beta_j/\alpha_j a_j$, i.e. $-\beta_j/\alpha_j =: \gamma_j$ 
are the eigenvalues of $X := JG$ with respect to the eigenvectors $a_j$. Hence we make the following
\begin{conv} \label{conv:ev} \hfill
  \begin{enumerate}
  \item Let $\{a_i, i \in \overline{2n}\}$ be a basis of eigenvectors
    of $M$. We will assume that this system is fixed and relabeled in
    such a way that for $i \in \overline{n}$ it holds
    $\lambda_i = \bar \lambda_{n + i}$ with
    $\mathrm{Im}(\lambda_i) > 0$.
  \item The eigenvalue of $X_k = J G_k$ for $k = 1, 2$ with respect to the
    eigenvector $a_i$ of $M$ (labeled according to point 1) is denoted by
    $\gamma_i^{(k)}$.
  \end{enumerate}
\end{conv} \noindent
Since $b_j = J a_j$, Eq. \eqref{eq:ab}
provides us with the following relation between the norm of $v_i$,
$\alpha_i$ and $a_i$:
\begin{equation}
  \langle v_i, v_i \rangle = - |\alpha_i|^2 \gamma_i \langle a_i, J a_i \rangle . \label{eq:viai}
\end{equation}
\begin{rmk} \label{rmk:giimag} By Eq. \eqref{eq:viai} it follows in
  particular, since
  $\overline{\langle a_i, J a_i \rangle} = a_i^{tr} J \bar a_i = -
  a_i^H J a_i = - \langle a_i, J a_i \rangle$,
  that the eigenvalues $\gamma_i$ must be purely imaginary.
\end{rmk} \noindent By assumption it holds 
$G = S^{tr} D S$,
so $X = JG = J S^{tr} D S = S^{-1} JD S$ and therefore $JD$ has the
same eigenvalues $\gamma_i$ as $X$.  A unitary basis $\{f_j; j \in \overline{2n}\}$
of eigenvectors of $JD$ is given as follows:

\begin{dfnsz}  \label{dfnsz:eig_jd} 
Denote by $\{e_j, j \in \overline{2n}\}$ the canonical basis
of $\mathbb{C}^{2n}$. For $j \in \overline{2n}$ define
\begin{align*}
  \sqrt{2} f_j := \left\{\begin{array}{cl} 
                           e_j + i e_{n + j} & \text{if } j \in \overline{n} , \\ 
                           e_j - i e_{n + j} & \text{else}.  \end{array}\right.
\end{align*}
Then it holds:
  \begin{enumerate}
  \item $\{f_j\}_{j \in \overline{2n}}$ span a unitary basis of
    $\mathbb{C}^{2n}$.
  \item Let $D = \mathrm{diag}(\Lambda, \Lambda) \in \mathbb{C}^{2n \times 2n}$ block-diagonal with
    diagonal $n \times n$-matrices
    $\Lambda := \mathrm{diag}(\Lambda_1, ..., \Lambda_n)$. Then
    \[
    JD f_j = \left\{\begin{array}{cl} i \Lambda_j f_j & \text{if } j
        \in \overline{n} , \\ - i \Lambda_j f_j &
        \text{else.} \end{array}\right.
    \]
  \item $Jf_j = i f_j$ for $j \in \overline{n}$ and $J f_j = - i f_j$
    else.
  \item $D f_j = \Lambda_j f_j$.
  \end{enumerate}
\end{dfnsz}
\begin{proof}
  \begin{enumerate}
  \item If $j, k \in \overline{n}$, then
    $2 \langle f_j, f_k \rangle = \langle e_j + i e_{n + j}, e_k + i
    e_{n + k} \rangle = \langle e_j, e_k \rangle + \langle e_{n + j},
    e_{n + k} \rangle = 2 \delta_{jk}$; similarly is the case $j, k \in \overline{2n} \backslash \overline{n}$.
    If $j \in \overline{n}$ and
    $k \in \overline{2n} \backslash \overline{n}$, then
    $2 \langle f_j, f_k \rangle = \langle e_j + i e_{n + j}, e_k - i
    e_{n + k} \rangle = \langle e_j, e_k \rangle - \langle e_{n + j},
    e_{n + k} \rangle = 0$.
  \item
    \begin{equation}
      JD f_j = \left(
        \begin{array}{cc} 
          0 & \Lambda \\
          - \Lambda & 0 
        \end{array}
      \right) f_j 
      = \frac{1}{\sqrt{2}} (- \Lambda_j e_{n + j} \pm i \Lambda_j e_j) =
      \pm i \Lambda_j f_j . \label{eq:jdf_j}
    \end{equation}
  \item This follows from Eq. \eqref{eq:jdf_j} by setting
    $\Lambda = 1_n$.
  \item $Df_j = -J^2 Df_j = \mp i \Lambda_j J f_j = \Lambda_j f_j$.
  \end{enumerate}
\end{proof}

\begin{dfnsz} \label{dfnsz:strns} For
$a, b \in \overline{n}$
define
  an orthosymplectic transposition $T_{ab} \in \mathbb{C}^{2n}$ as
  follows:
  \[
  T_{ab}(e_j) := \left\{
    \begin{array}{cl}
      -e_{n + b} & \text{if } j = a , \\
      e_b & \text{if } j = n + a , \\
      -e_{n + a} & \text{if } j = b , \\
      e_a & \text{if } j = n + b , \\
      e_j & \text{else.}
    \end{array}\right.
  \]
  Then it holds
  \begin{enumerate}
  \item $T_{ab}$ is orthogonal and symplectic.
  \item If $D' := T_{ab}^{tr} D T_{ab}$ then $D'$ has the same form of
    $D$ where $\Lambda_a$ is exchanged with $\Lambda_b$.
  \end{enumerate}
\end{dfnsz}
\begin{proof}
  $T_{ab}$ is orthogonal, since
  $T_{ab}$ is a combination of permutation and a reflection which are
  orthogonal.  The symplecticity and the second property can be seen
  as follows: Without loss of generality we consider only indices
  $j \in \{a, b, n + a, n + b\}$. Since $T_{ab}$ is orthogonal,
  $T_{ab}^{tr} = T_{ab}^{-1}$. Then
  \begin{align*}
    T^{tr}_{ab} J T_{ab}(e_j) &=
                                \left\{
                                \begin{array}{cl}
                                  -T_{ab}^{-1} J (e_{n + b}) = - T^{-1}_{ab}(e_b) = -e_{n + a} = J(e_a) & \text{if } j = a , \\
                                  T_{ab}^{-1} J (e_b) = - T_{ab}^{-1}(e_{n + b}) = e_a = J(e_{n + a}) & \text{if } j = n + a , \\
                                  - T_{ab}^{-1} J(e_{n + a}) = - T_{ab}^{-1}(e_a) = - e_{n + b} = J(e_b) & \text{if } j = b , \\
                                  T_{ab}^{-1} J (e_a) = - T_{ab}^{-1}(e_{n+ a}) = e_b = J(e_{n+ b}) & \text{if } j = n + b .
                                \end{array}\right\} = J(e_j). \\
    T^{tr}_{ab} D T_{ab}(e_j) &=
                                \left\{
                                \begin{array}{cl}
                                  -T_{ab}^{-1} D (e_{n + b}) = \Lambda_b e_a & \text{if } j = a , \\
                                  T_{ab}^{-1} D (e_b) = \Lambda_b e_{n + a} & \text{if } j = n + a , \\
                                  - T_{ab}^{-1} D (e_{n + a}) = \Lambda_a e_b & \text{if } j = b , \\
                                  T_{ab}^{-1} D (e_a) = \Lambda_a e_{n + b} & \text{if } j = n + b.
                                \end{array}\right.
  \end{align*}
\end{proof}

Now attach on the matrices the index 1 and 2. Since the eigenvalues of
$JD_k$ are the $\gamma^{(k)}_j$'s and the
$\Lambda^{(k)}_j$'s are positive, there must exist, by
Conv. \ref{conv:ev}, a permutation
$\pi_k \colon \overline{n} \to \overline{n}$ so that
$\gamma^{(k)}_j = i \Lambda^{(k)}_{\pi_k(j)}$ for $j \in \overline{n}$
. Note that this implies automatically
$\gamma^{(k)}_{n + j} = - i \Lambda^{(k)}_{\pi_k(j)}$ for
$j \in \overline{n}$, since the complex conjugated $\overline{\gamma_j^{(k)}}$ belongs to the eigenvector
$S \bar a_j$ of $JD_k$, which in turn equals $S a_{n + j}$ by 
our Conv. \ref{conv:ev} and therefore it is related to the eigenvalue $\gamma_{n + j}^{(k)}$.\\\\
By Prop. \ref{dfnsz:strns}, we can assign to
$\pi_k$ a suitable composition $T_k$ of symplectic permutation
matrices so that the indices of the diagonal entries of
$D'_k := T_k^{tr} D_k T_k$ now coincide with the indices $j$ of
$\gamma_j^{(k)}$, with respect to our fixed eigensystem
$\{a_i, i \in \overline{2n}\}$.  These diagonal matrices $D'_k$ belong
to a similar problem than the original one, now formulated with the
symplectic matrices $\tilde S_k := T_k^{-1} S_k$ (and Prop. \ref{dfnsz:eig_jd} holds also for these
new block-diagonal matrices).  Since $T_k$ is
orthogonal, $S_2 S_1^{-1}$ is orthogonal iff
$\tilde S_2 \tilde S_1^{-1}$ is orthogonal.  The important fact of
this consideration is that we treat both cases $k = 1, 2$
simultaneously (if we would have looked at only one case, we could have
simply relabeled the $a_i$'s). So we conclude:

\begin{cor}
  Without loss of generality we can assume that for all
  $j \in \overline{n}$ it holds $\gamma_j^{(k)} = i \Lambda^{(k)}_j$
  and $\gamma_{n + j}^{(k)} = - i \Lambda^{(k)}_j$, i.e.
  \begin{equation}
    \forall k \in \{1, 2\} \colon \forall j \in \overline{2n} \colon \;\; JD_k f_j = \gamma^{(k)}_j f_j , \label{eq:jdfj_fj}
  \end{equation}
  in particular $f_j \in E^{(k)}_j$, by which we denote the eigenspace of $JD_k$ with respect to $\gamma_j^{(k)}$, and
  $\{f_j; j \in \overline{2n}\}$ is a unitary basis of eigenvectors of
  $JD_k$.
\end{cor}

\subsection{Proof of Thm. \ref{thm:mt} Part 2} \label{subsec:evev2} 

\begin{prop} \label{prop:evjd} By assumption we have
  $G = S^{tr} D S = (D^{1/2} S)^{tr} D^{1/2} S$, therefore we can
  apply the results of paragraph \ref{subsec:evev}, using in particular
  $P = D^{1/2} S$. Then the corresponding orthogonal eigenvectors
  $c_j$ of $PMP^{-1}$ (with this particular $P$) are also eigenvectors
  of $JD$ and satisfy:
  \begin{equation}
    - \gamma_j J D^{-1} c_j = c_j \;\; \Rightarrow \;\; JD c_j = \gamma_j c_j .  \label{eq:jdcj_cj}
  \end{equation}
\end{prop}
\begin{proof}
  By Eq. \eqref{eq:ab}, using $b_j = Ja_j$, we have with
  corresponding $\alpha_i$ and $\beta_i$ values (they are not yet
  specified)
  \begin{align*}
    \langle c_i, c_j \rangle &= \bar \alpha_i \beta_j \langle a_i, J a_j \rangle = \bar \alpha_i \beta_j \langle S a_i, J S a_j \rangle \\
                             &= \langle \alpha_i S a_i, \beta_j J S a_j \rangle = \langle D^{-1/2} c_i, \beta_j / \alpha_j J D^{-1/2} c_j \rangle \\
                             &= \beta_j / \alpha_j \langle c_i, JD^{-1} c_j \rangle .
  \end{align*}
  Since the $c_i$'s constitute a basis, the claim follows.
\end{proof}
By Prop. \ref{prop:evjd} and paragraph \ref{subsec:evev} we thus have for each $k$ an
orthogonal basis $\{c_j^{(k)}\}_{j \in \overline{2n}}$ of eigenvectors
of $JD_k$, satisfying $c_j^{(k)} = \alpha_j^{(k)} D_k^{1/2} S_k(a_j)$
for not yet specified complex numbers
$\alpha_j^{(k)} \in \mathbb{C} \backslash \{0\}$.  Let us now choose
$\alpha_j^{(k)}$ so that $c_j^{(k)}$ are normalized to one, i.e. they
describe a unitary basis.  By Eq. \eqref{eq:viai} this is fulfilled if
and only if
\begin{equation}
  \forall i \in \{1, ..., 2n\} \colon \;\; - |\alpha_i^{(k)}|^2 \gamma_i^{(k)} \langle a_i, J a_i \rangle = 1 ,
\end{equation}
leaving an $SU(2)$ freedom in the choice of the $\alpha_i^{(k)}$'s. In
particular we obtain
\begin{equation}
  \forall i \in \{1, ..., 2n\} \colon \;\; \left| \frac{\alpha_i^{(1)}}{\alpha_i^{(2)}} \right|^2 = \frac{\gamma_i^{(2)}}{\gamma_i^{(1)}} . \label{eq:ai1ai2}
\end{equation}
Now let $U_k \in \mathbb{C}^{2n \times 2n}$ be the unitary
transformation sending $c^{(k)}_j$ to $f_j$.  Let us drop the index
$k$ for the next lemma.
\begin{lem} \label{lem:uju} The unitary map $U$ satisfies
  \begin{enumerate}
  \item $U JD U^H = J D$.
  \item $U D^{\pm 1/2} U^H = D^{\pm 1/2}$.
  \item $U D^{\pm 1} U^H = D^{\pm 1}$.
  \item $U J U^H = J$.
  \end{enumerate}
\end{lem}
\begin{proof}
  \begin{enumerate}
  \item Since $JD c_j = \gamma_j c_j$ by Eq. \eqref{eq:jdcj_cj}, we
    have by Eq. \eqref{eq:jdfj_fj}:
    $U JD U^H(f_j) = \gamma_j f_j = JD f_j$.
  \item Let $E_j$ be the eigenspace of $JD$ with respect to $\gamma_j$
    and $[j]$ the equivalence class of indices $k \in \overline{2n}$
    with $k \sim j : \Leftrightarrow \gamma_k = \gamma_j$. Since
    $\forall i \in [j] \colon \; U^H(f_i) = c_i \in E_j$, and
    $\{f_i \in E_j ; i \in [j]\}$ is a basis of $E_j$, it follows that
    there exist $u_{il} \in \mathbb{C}$ so that
    \[
    U^H(f_i) = \sum_{l \in [j]} u_{il} f_l .
    \]
    Then
    \begin{align*}
      & D^{\pm 1/2} U^H(f_i) = \sum_{l \in [j]} u_{il} |\gamma_l|^{\pm 1/2} f_l = |\gamma_j|^{\pm 1/2} \sum_{l \in [j]} u_{il} f_l = |\gamma_j|^{\pm 1/2} U^H (f_i) , \\
      \Rightarrow \;\; & U D^{\pm 1/2} U^H = D^{\pm 1/2} . 
    \end{align*}
  \item Follows immediately from 2.
  \item By 1 and 3:
    $JD = UJU^HUDU^H = UJU^HD \;\; \Rightarrow \;\; J = UJU^H$.
  \end{enumerate}
\end{proof}
Now we are ready to prove the original claim. Set $P_k := U_k D_k^{1/2} S_k$,
i.e.  $f_j = \alpha^{(k)}_j P_k(a_j)$. It follows
\begin{equation}
  U_2 D_2^{1/2} S_2 S_1^{-1} D_1^{-1/2} U_1^H f_j = P_2 P_1^{-1} f_j = \alpha_j^{(1)} P_2 a_j = \frac{\alpha_j^{(1)}}{\alpha_j^{(2)}} f_j . \label{eq:dssd}
\end{equation}
By lemma \ref{lem:uju} it holds $U_2 D_2^{1/2} = D_2^{1/2} U_2$ and
$D_1^{-1/2} U_1^H = U_1^H D_1^{-1/2}$, and therefore with
Eq. \eqref{eq:ai1ai2}
\begin{align*}
  D_2^{1/2} U_2 S_2 S_1^{-1} U_1^H D_1^{-1/2} f_j &= \frac{\alpha_j^{(1)}}{\alpha_j^{(2)}} f_j , \\
  \Rightarrow \;\; U_2 S_2 S_1^{-1} U_1^H (|\gamma_j^{(1)}|^{-1/2} f_j) &= \frac{\alpha_j^{(1)}}{\alpha_j^{(2)}} |\gamma_j^{(2)}|^{-1/2} f_j , \\
  \Rightarrow \;\; U_2 S_2 S_1^{-1} U_1^H (f_j) &= \frac{\alpha_j^{(1)}}{\alpha_j^{(2)}} \left| \frac{\alpha_j^{(2)}}{\alpha_j^{(1)}} \right| f_j .
\end{align*}
Since multiplication of a unitary basis with complex phases is a
unitary operation, $S_2 S_1^{-1}$ can entirely be described on
$\mathbb{C}^{2n}$ as a unitary operation. And because $S_1$ and $S_2$ itself
are real, we conclude that $S_2 S_1^{-1}$ must be orthogonal. $\Box$

\bibliographystyle{unsrt}
\bibliography{references}

\end{document}